\documentclass[reqno,11pt]{amsart}
\usepackage[linktocpage]{hyperref}
\usepackage[utf8]{inputenc}
\usepackage{graphicx}
\usepackage{slashed}
\usepackage{amscd}
\usepackage{amssymb}
\usepackage{esint}
\usepackage{enumitem}
\usepackage{comment}
\usepackage{amsmath}
\usepackage{dsfont}
\usepackage[mathscr]{eucal}
\usepackage{xcolor}
\usepackage{pst-all}

\numberwithin{equation}{section}
\allowdisplaybreaks[1]
\definecolor{labelkey}{gray}{.65}

\title[Constraints for Physical Gauge Groups]{Constraints for Physical Gauge Groups \\
Coming from the Causal Action Principle}

\author[F.\ Finster]{Felix Finster}
\address{Fakult\"at f\"ur Mathematik \\ Universit\"at Regensburg \\ D-93040 Regensburg \\ Germany}
\email{finster@ur.de}

\author[N.G.\ Gresnigt]{Niels G. Gresnigt \\ \\ September 2026}
\address{Department of Physics, Xi'an Jiaotong-Liverpool University \\ Suzhou \\ China}
\email{niels.gresnigt@xjtlu.edu.cn}

\newtheorem{Def}{Definition}[section]
\newtheorem{Thm}[Def]{Theorem}
\newtheorem{Prp}[Def]{Proposition}
\newtheorem{Lemma}[Def]{Lemma}
\newtheorem{Remark}[Def]{Remark}
\newtheorem{Corollary}[Def]{Corollary}
\newtheorem{Example}[Def]{Example}

\newcommand{\Thanks}{\vspace*{.5em} \noindent \thanks}
\newcommand{\beq}{\begin{equation}}
\newcommand{\eeq}{\end{equation}}
\newcommand{\Proof}{\begin{proof}}
	\newcommand{\QED}{\end{proof} \noindent}
\newcommand{\QEDrem}{\ \hfill $\Diamond$}

\newcommand{\la}{\langle}
\newcommand{\ra}{\rangle}

\newcommand{\C}{\mathbb{C}}
\newcommand{\R}{\mathbb{R}}
\newcommand{\1}{\mathbb{I}}

\newcommand{\N}{\mathbb{N}}
\newcommand{\Pdd}{\mbox{$\partial$ \hspace{-1.2 em} $/$}}

\newcommand{\V}{\mathcal{V}}

\renewcommand{\Tr}{\text{\rm{Tr}}}

\newcommand{\U}{\text{\rm{U}}}

\newcommand{\SU}{\text{\rm{SU}}}

\renewcommand{\u}{\text{\rm{u}}}
\newcommand{\g}{{\mathfrak{g}}}

\renewcommand{\H}{\mathscr{H}}

\DeclareMathOperator{\Symm}{\mbox{\rm{Symm}}}

\newcommand{\I}{\bb{I}}

\newcommand{\scrM}{\mycal M}

\newcommand{\bitem}{\begin{itemize}[leftmargin=2.5em]}
\newcommand{\eitem}{\end{itemize}}

\newcommand{\G}{{\mathscr{G}}}

\newcommand{\Comm}{\text{\rm{Comm}}}
\newcommand{\Span}{\operatorname{span}}
\newcommand{\intprod}{\mathbin{\lrcorner}}
\newcommand{\extprod}{\wedge}
\newcommand{\pseudo}{\Gamma}

\DeclareFontFamily{OT1}{rsfso}{}
\DeclareFontShape{OT1}{rsfso}{m}{n}{ <-7> rsfso5 <7-10> rsfso7 <10-> rsfso10}{}
\DeclareMathAlphabet{\mycal}{OT1}{rsfso}{m}{n}

\newcommand{\bb}[1]{\mathbb{#1}}
\newcommand{\Cl}{C\ell}

\begin{document}
\maketitle

\begin{abstract}
The constraints for the effective local gauge groups stemming from the causal  action principle for causal fermion systems are reviewed. The constraint which are quadratic in the gauge potentials (coming from the so-called bilinear logarithmic terms) are shown to make a connection between the structures of Lie algebras and Clifford algebras. This connection is worked out systematically starting from general chiral potentials corresponding to the gauge group $\text{\rm{U}}(N) \times \text{\rm{U}}(N)$. The general results are illustrated in various examples.
\end{abstract}

\tableofcontents

\section{Introduction} \label{secintro}
The theory of {\em{causal fermion systems}} is a recent approach to fundamental physics
(for an introduction to the physical background and applications,
we refer the interested reader to the review~\cite{review}, the textbooks~\cite{cfs, intro} or the website~\cite{cfsweblink}). In this approach, the physical equations are formulated
in terms of a novel variational principle in spacetime, the {\em{causal action principle}}.
In~\cite[Chapters~3--5]{cfs} it is shown that, taking the vacuum configuration of the
fermions in the standard model as the starting point, the causal action principle
determines the structure of the interaction in detail, giving the
interactions of the standard model after spontaneous symmetry breaking.
More precisely, the gauge group~$\U(1) \times \SU(2) \times \SU(3)$ is derived, and it is
shown that the corresponding gauge potentials couple to the fermions exactly as
the electroweak and strong potentials in the standard model. In particular,
the~$\SU(2)$ gauge potentials are shown to be left-handed and massive.

These results were obtained by a detailed analysis of the eigenvalues of the operator
product~$xy$ (for details see~\cite[Chapters~3--5]{cfs} or the summary in
Section~\ref{secconstraints}). An important part of this analysis is concerned with
the question whether and to which extent the degeneracies of these eigenvalues
are removed by the gauge potentials and fields.
In order to get a deeper and more systematic understanding of the underlying
mechanisms, it is an important task to study the effective gauge fields for more
general vacuum configurations. The present paper is a first important step towards
this goal. Our main result is to formulate the core of the degeneracy analysis
in~\cite[Chapters~3--5]{cfs} for arbitrary gauge potentials acting on general
vector bundles.
The resulting conditions can be solved explicitly, giving constraints
 for the possible choices of the effective gauge groups. This general
structural result stated in Theorem~\ref{thm:main} will be illustrated in various examples.

Let us be more specific on what we mean by the ``core of the degeneracy analysis''.
Some of the constraints for the gauge potentials derived in~\cite[Chapters~3--5]{cfs} 
can be described simply by linear equations (see Sections~\ref{secchiral} and~\ref{secfield}).
The more difficult and interesting constraint (coming from the so-called
bilinear logarithmic terms; see Section~\ref{secbillog}) is quadratic in the
chiral potentials. Here we focus on this constraint and simplify the problem
by assuming that the mass matrix is invertible and commutes with all chiral potentials.
With this simplification, we disregard a mixing of the generations
(as described in the standard model by the CKM and PMNS matrices; we leave
the extension of our methods and results to include mixing matrices as a project
for future research). After this simplification, the constraint can be formulated
as a general condition relating the structures of Lie algebras and Clifford algebras.
After formulating the constraint in this language, we give a general classification
theorem for Lie algebras satisfying this constraint.

Our methods and results are also of broader interest in the context of programs which seek to derive the gauge symmetry of the standard model from more fundamental structures. In conventional grand unified theories, the standard model gauge group is embedded into a larger group and recovered through spontaneous symmetry breaking~\cite{ross,baez2010algebra}. Possible breaking patterns and particle multiplets are then analyzed using subgroup chains and branching rules. These methods are powerful, but do not by themselves single out a unique low-energy gauge group.

Other approaches attempt to characterize the gauge group more intrinsically. These include non-commutative geometry~\cite{chamseddine2008why}; division-algebraic, octonionic and sedenionic constructions~\cite{singh-octo2,furey2022division,Gresnigt2023}; Clifford-algebraic and spinorial approaches~\cite{stoica2018leptons,gording2020unified,wilson2021problem,gourlay2024algebraic,quinta2025spacetime,Gresnigt2026PLB,krasnov2024geometry}; and constructions based on the exceptional Jordan algebra~\cite{todorov2018deducing,baez2026standard}. Depending on the framework, the gauge group may be identified as an automorphism group, a stabilizer, a commutant, an intersection of distinguished subgroups or through related structural criteria. The methods developed here belong to this broader program by providing a way of selecting admissible effective gauge algebras through intrinsic properties of the Lie algebras and their action on the relevant Clifford subspaces, without choosing a symmetry-breaking chain a priori.

We now introduce and state the core of our mathematical result.
Given~$N>1$, we begin with the Lie algebra~$\u(N) \oplus \u(N)$.
Thus the elements in this Lie algebra are pairs
denoted by~$(A_L, A_R) \in \u(N) \oplus \u(N)$.
The reader interested in the physical context can think of the indices~$L$ and~$R$ 
as labeling the potentials which act on the left- and right-handed components of the
fermions, respectively. But our results clearly hold independently of the physical context.
Next, we let~$\g$ be a Lie subalgebra of~$\u(N) \oplus \u(N)$.
We impose the {\em{Clifford compatibility condition}} which states
that all elements~$(A_L, A_R) \in \g$ of this subalgebra should have
the property that the difference of its two components squares to a multiple of the identity
matrix, i.e.\
\beq \label{liecond}
(A_L - A_R)^2 = c\:\I \qquad \text{with} \qquad c=c(A_L, A_R) \in \C \:.
\eeq
Moreover, we demand that~$\g$ is {\em{maximal}} with respect to this condition,
meaning that there is no proper Lie algebra extension~$\g' \supset \g$
which also satisfies the Clifford compatibility condition~\eqref{liecond} for
all~$(A_L, A_R) \in \g'$.

We note that the Clifford compatibility condition poses quadratic constraints for the
form of all Lie algebra elements. Our main task is to understand what all these
conditions mean, and how the resulting admissible Clifford algebras look like.
The name Clifford compatibility condition is inspired by the fact that~\eqref{liecond}
gives a direct {\em{connection to Clifford algebras}}. Indeed, setting
\[ \mathscr{A} := \text{span} \big\{ A_L - A_R \:\big|\: (A_L, A_R) \in \g \big\} \subset
\Symm(\C^N) \]
(where~$\Symm(\C^N)$ denotes all Hermitian $N \times N$-matrices),
we know that~$w^2 = c \1$ for all~$w \in \mathscr{A}$. Evaluating this equation
for the linear combination~$w = \alpha u + \beta v$ with~$\alpha, \beta \in \R$
and multiplying out,
\[ c \,\1 = (\alpha u + \beta v)^2 = \alpha^2 u^2 + \beta^2 v^2 +
2 \alpha \beta\: \{u, v \} \:, \]
one finds that the anti-commutator of~$u$ and~$v$ is also a multiple of the
identity. This standard argument, sometimes referred to as the {\em{polarization
formula}}, shows that on~$\mathscr{A}$ the Clifford anti-commutation relations hold,
\beq \label{sproddef}
\frac{1}{2}\: \{u, v\} = \la u,v \ra \:,
\eeq
where~$\la .,. \ra$ is a scalar product on~$\mathscr{A}$ (we can take~\eqref{sproddef}
as the definition of this scalar product). Denoting the real Clifford algebra
generated by~$\mathscr{A}$ by~$\Cl(\mathscr{A})$ and identifying~$\mathscr{A}= \Cl^{[1]}(\mathscr{A})$ with the grade one
subspace, we thus obtain a mapping from a Lie algebra to a Clifford algebra,
\[ \g \mapsto \Cl(\mathscr{A}) \:,\qquad (A_L, A_R) \mapsto (A_L - A_R) \in \mathscr{A} = \Cl^{[1]}(\mathscr{A})\:. \]
The existence of this mapping poses interesting constraints connecting the structures
of Lie and Clifford algebras. More specifically, the commutator of the
Lie bracket must be compatible with the anti-commutator of the Clifford algebra,
meaning that for all~$(A_L, A_R), (A_L', A_R') \in \g$,
\[ i [A_L, A'_L] - i [A_R, A'_R] \in \mathscr{A} \:. \]
Our main result gives a complete classification of all Lie algebras
which satisfy all these constraints.
Before stating this result, we need to introduce some more notation.
We denote the commutant of~$\mathscr{A}$ by
\[ \Comm(\mathscr{A}) := \big\{ u \in \Symm(\mathscr{A}) \:\big|\: [u,v]=0 \text{ for all~$v \in \mathscr{A}$} \big\} \:. \]
The {\em{complex}} Clifford algebra generated by~$\mathscr{A}$ is denoted
by~$\bb{C}\ell(\mathscr{A})$. We denote the grade two subspace by~$\bb{C}\ell^{[2]}(\mathscr{A})$
(it is sometimes also referred to as the space of bilinear covariants).
Finally, denoting the dimension of~$\mathscr{A}$ by~$p$,
the {\em{pseudoscalar operator}}~$\omega$ is introduced by
\[ \omega := i^{\lfloor p/2\rfloor}\,\gamma_1\gamma_2\cdots\gamma_p \:. \]
We remark for clarity that, in the case that~$p$ is odd, the pseudoscalar
operator commutes with all vectors in~$\mathscr{A}$, so that
\[ \omega \in \Comm(\mathscr{A}) \qquad \text{if~$p$ is odd}\:. \]
For simplicity, we here state our main theorem only in the case that~$p$ is
even; the statement in the case that~$p$ is a bit more involved and will be given in
Theorem~\ref{thm:main}.
\begin{Thm} \label{thmintro} Let~$\g$ be a maximal Lie algebra satisfying the
Clifford compatibility conditions~\eqref{liecond}. Assume that
the dimension~$p$ of the Clifford algebra~$\mathscr{A}$ is odd.
Then all potentials~$(A_L, A_R)$ in~$\g$ are of the form
\[ A_L = \gamma+ \sigma + C \:,\qquad A_R = -\gamma+\sigma+C \]
with
\[ \gamma\in \mathscr{A},\qquad\sigma\in\bb{C}\ell^{[2]}(\mathscr{A}),\qquad C\in\Comm(\mathscr{A})\:. \]
\end{Thm} \noindent
This theorem will be proved in Section~\ref{secmain} (see Theorem~\ref{thm:main}).
We also illustrate this result in various examples and work out applications
(see Section~\ref{secexamples}).

The paper is organized as follows. Section~\ref{secprelim}
provides the necessary background on causal fermion systems
(the Clifford algebra concepts used in the mathematical analysis
are introduced in Section~\ref{secmain} as they are needed).
In Section~\ref{secconstraints}, we review the constraints for the
effective gauge potentials obtained in~\cite[Chapters~3--5]{cfs} in a concise
mathematical way. In Section~\ref{secmain} we make the general connection
to Lie and Clifford structures and prove our main theorem (Theorem~\ref{thm:main}).
Finally, in Section~\ref{secexamples} we work
out and discuss various examples.

\section{Preliminaries on Causal Fermion Systems in Minkowski Space} \label{secprelim}
For brevity, we do not give a general introduction to causal fermion systems
and the causal action principle, referring the reader interested in the
general context to the textbooks~\cite{intro, cfs} or reviews like~\cite{rrev, dice2014}.
Instead, we focus on those structures and results needed for our analysis.
Restricting attention to causal fermion systems describing classical spacetimes
and disregarding gravity, the physical data encoded in the causal fermion system
consists of a family of spinorial wave functions in Minkowski space.
Out of this family, one constructs a bi-distribution~$P(x,y)$, referred to as the
{\em{kernel of the fermionic projector}}. The causal action principle can be formulated
in terms of~$P(x,y)$ and gives many conditions and constraints for the form
of~$P(x,y)$. Translated in terms of the individual wave functions, these conditions
imply that the wave functions must satisfy the Dirac equation coupled to
classical fields, which in turn must satisfy corresponding classical field
equations\footnote{The causal action principle also gives rise to
an interaction via quantum fields (see~\cite{fockdynamics} or the review~\cite{qftlimit}), but this will not be considered here.}.
This procedure is dubbed {\em{continuum limit analysis}}.
As input to this procedure, one must prescribe the vacuum configuration of the
wave functions. As will be described in more detail below, this is done by
building up the kernel of the fermionic projector from Dirac sea configurations,
being solutions of the vacuum Dirac equation.
Once this has been done, the causal action principle determines the
interaction completely. In particular, the continuum limit analysis
tells us which gauge fields appear in the theory and how they couple to
the fermions. We refer to the potentials of these gauge fields as the
{\em{dynamical gauge potentials}}.

We proceed by explaining how to build up vacuum configurations.
Our construction does not only allow for the description of the standard model,
but includes any configuration involving left- and right-handed Dirac spinors.
Thus it is general enough for modeling simple toy models as well as
extensions of the standard model.

\subsubsection{Describing the Vacuum in Minkowski space}
We assume familiarity with the Dirac equation in Minkowski space
(for basics see for example~\cite[Sections~1.1--1.3]{intro} or standard
textbooks like~\cite{bjorken, itzykson/zuber, peskin+schroeder}).
The basic building block for the kernel of the fermionic projector of the
vacuum is the distribution~$P_m(x,y)$ defined by
\[ P^\text{vac}_m(x,y) = \int \frac{d^4k}{(2 \pi)^4}\: (\gamma^j k_j +m)\: \delta(k^2-m^2)\: \Theta(-k^0)\: e^{-ik(x-y)}\:. \]
Here~$x,y \in \scrM = \R^{1,3}$ are spacetime points in Minkowski space. The variable~$k$ denotes the momentum, and~$k(x-y)$ denotes the Minkowski
inner product. Moreover, $\gamma^j$ are the Dirac matrices.
By direct computation, one verifies that~$P^\text{vac}_m(x,y)$ satisfies the
Dirac equation of mass~$m$,
\[ (i \Pdd_x - m)\, P^\text{vac}_m(x,y) = 0 \:. \]
We point out that, due to the Heaviside function~$\Theta(-k^0)$,
the distribution~$P^\text{vac}(x,y)$ is built up of all {\em{negative frequency
solutions}} of the Dirac equation. This implements the physical picture
that, in the vacuum, all one-particle states of negative frequency are occupied
forming the so-called {\em{Dirac sea}}. We refer the reader interested in the
physical background to~\cite[Section~1.5]{intro}.

The kernel of the fermionic projector of the vacuum~$P^\text{vac}(x,y)$
is built up of sums and direct sums of Dirac seas. More precisely,
generalizing the procedure in~\cite{cfs}, we set
\beq \label{Pgen}
P^\text{vac}(x,y) = \bigoplus_{a=1}^N \sum_\beta X_{a \beta}
P_{m_{a \beta}}(x,y) \:.
\eeq
We refer to the direct summands as {\em{sectors}}; their number is denoted
by~$N \in \N$. In each sector, we have a finite sum over Dirac seas
with masses~$m_{a \beta}$. These summands describe the {\em{generations}}
of elementary particles in this sector. In the standard model, the number of
generations is three (describing for example electron, muon and tau leptons).
In general, the number of generations can be any number, and it may also be
different in different sectors. For this reason, we did not specify the limits of the
sum in~\eqref{Pgen}. Finally, the factors~$X_{a \beta}$ are $4 \times 4$-matrices
which describe a chiral asymmetry (implementing for example that neutrinos
are left-handed). They can be chosen from the following list,
\[ X_{a \beta} = \left\{ \begin{array}{ll}
\1 & \text{no chiral asymmetry} \\[0.1em]
\1 + \tau_\text{reg} \chi_L & \text{left-handed asymmetry} \\[0.1em]
\1 + \tau_\text{reg} \chi_R & \text{right-handed asymmetry}\:.
\end{array} \right. \]
Here~$\tau_\text{reg}$ is a real parameter which is
useful for keeping track of chiral asymmetry in composite expressions.
Moreover, the {\em{chiral projectors}}~$\chi_{L\!/\!R}$ are defined by
\[     \chi_L = \frac{1}{2} \left(\1 - \pseudo
 \right) \qquad \text{with} \qquad
 \pseudo = i \gamma^0 \gamma^1 \gamma^2 \gamma^3 \]
 ($\Gamma$ is referred to as the {\em{pseudoscalar matrix}}; in the physics
 literature, it is often denoted by~$\gamma^5$).

\subsubsection{The Chiral Gauge Group and Chiral Gauge Potentials} \label{sechirgauge}
We now simplify the above setting by restricting attention to 
the case of {\em{one generation}} of elementary particles.
This will provide the first step in the systematic analysis of the effective gauge potentials.
We plan to consider the case of several generations in a separate publication. Thus we simplify the general ansatz~\eqref{Pgen} to
\[ P^\text{vac}(x,y) = \bigoplus_{a=1}^N X_a P_{m_a}(x,y) \:. \]

We form the {\em{chiral gauge group}}~$\G$ by all pairs of unitary transformations
acting on the left- and right-handed component of the sectors in the fundamental representation.
Thus
\[ \G := \U(N)_L \times \U(N)_R \:, \]
in the representation
\[ 
(U_L, U_R) \psi = \bigoplus_{a=1}^N \sum_{b=1}^N \Big( (U_L)^a_b\, \chi_L\, \psi_b
+ (U_R)^a_b\, \chi_R\, \psi_b \Big) \:. \]
The {\em{chiral potentials}} are sections in the corresponding Lie algebra, i.e.\
\[ (A_L^j, A_R^j) \in \g := \u(N) \oplus \u(N) \:. \]
The analysis in~\cite[Chapters~3--5]{cfs} gives constraints for the chiral potentials.
This will be summarized in Section~\ref{secconstraints}.

\section{Constraints for the Chiral Gauge Potentials} \label{secconstraints}
\subsection{Constraints from the Chiral Asymmetry} \label{secchiral}
The sectors on which the matrices~$X_L$ and/or~$X_R$ do not act like the identity
involve what we call
a {\em{non-trivial regularization}}. Here we do not need to be specific on what this
regularization looks like. Instead, by ``non-trivial'' we simply mean that no chiral gauge potentials
are allowed which mix the ``non-trivial'' sectors with the other sectors.
In mathematical terms, this means that the allowed chiral gauge potentials generate a Lie subgroup
of~$\G$ denoted by~$\G^\text{alg}$. This subgroup is formed of all unitary transformations which act like the identity on the sectors with non-trivial regularization,
i.e.\
\[ \G^\text{alg} := \big\{ (U_L, U_R) \in \G \:\big|\: U_L (\1-X_L) = (\1-X_L)\text{ and } U_R (\1-X_R) = (\1-X_R) \big\} \:. \]
Hence the corresponding algebraic gauge potentials~$(A_L, A_R)$ have the properties
\beq \label{Aalg}
A_L (\1-X_L) = 0 \qquad \text{and} \qquad A_R (\1-X_R) = 0 \:.
\eeq

\subsection{Constraints from the Bilinear Logarithmic Terms} \label{secbillog}
Following the procedure in~\cite[Section~5.3]{cfs}, we can derive further constraints for the
chiral gauge potentials. We refer to those potentials which satisfy all these constraints as the
{\em{admissible}} chiral gauge potentials. They form the {\em{admissible gauge group}}~$\G^\text{adm}$
(see Definition~\ref{defadm} below).

Compared to~\cite[Chapters~4 and~5]{cfs}, the setting here is considerably
simpler because we consider only one generation of elementary particles.
For clarity, we go through the constructions step by step and adapt the results to our setting.
We focus on the contributions to~$P(x,y)$ that have a logarithmic pole on the light cone.
The logarithmic pole is described in the formalism of the light cone expansion by the distribution~$T^{(1)}$.
The relevant contributions for our considerations are the {\em{logarithmic current term}}
(see~\cite[eq.~(B.2.7)]{cfs})
\[ 
\chi_L P(x,y) \asymp -\frac{1}{6}\: \chi_L \,j_L^i\,\gamma_i\: T^{(1)} \]
(by~$\asymp$ we denote a specific contribution to~$P(x,y)$),
and the {\em{logarithmic mass terms}} (see~\cite[eq.~(4.6.11)--(4.6.13)]{cfs})
\[ 
\chi_L P(x,y) \asymp \frac{1}{2}\: m^2\:Y Y \slashed{A}_L\: T^{(1)} -m^2\,\chi_L \:Y \slashed{A}_R \,Y\: T^{(1)} 
+\frac{1}{2}\:m^2\:\chi_L \:\slashed{A}_L \,Y Y \: T^{(1)} \:. \]
(We always give the formula for the left-handed component; the right-handed component is obtained by the
replacement~$L \leftrightarrow R$).
In order to satisfy the EL equations, the logarithmic poles need to be compensated for by
the {\em{microlocal chiral transformation}}. The corresponding contributions to the kernel of the fermionic
projector can be written as (see~\cite[Proposition~4.4.6]{cfs})
\beq \label{logmcl}
\chi_L P(x,y) \asymp -\chi_L\:\slashed{v}_L\: T^{(1)} \:.
\eeq
Evaluating the above contributions in the EL equations, they are contracted with a factor~$\xi$.
By a suitable choice of the vector field~$v_L$ in the microlocal chiral transformation,
the logarithmic pole of~$P(x,y)$ disappears. The remaining contributions by the current and mass terms
are smooth on the light cone. Comparing them with similar contributions by the Dirac current,
one obtains the Yang-Mills equations of the standard model.

Since we consider only one generation, the chiral gauge potentials by themselves do not give rise to any constraints (in contrast to the arguments in~\cite[\S5.3.2]{cfs}).
But in our setting it is important that the above contributions are modified by the chiral
gauge potentials. Expanding the resulting chiral gauge phases, one gets an expansion in
powers of~$A_L^i \xi_i$ and~$A_R^i \xi_i$. For ease of notation, we use the short
notation~$A_{L\!/\!R}[\xi] := A_{L\!/\!R}^i \xi_i$.
We restrict attention to the first-order term of this expansion. Since the resulting contributions to the Lagrangian
all involve two factors of~$\xi$, they are referred to as the {\em{bilinear logarithmic terms}}.

There are different contributions to the kernel of the fermionic projector which
all give rise to bilinear logarithmic terms.
The first contribution which is of relevance here is the
the term obtained inserting factors~$A_{L\!/\!R}[\xi]$ into the
{\em{logarithmic current term}} (see~\cite[eq.~(4.6.5)]{cfs})
\beq \label{log2j}
\chi_L P(x,y) \asymp
\frac{i}{12} \: \chi_L \left( A_L[\xi] \; j_L^i \gamma_i + j_L^i \gamma_i \;A_L[\xi] \right) \: T^{(1)} \:.
\eeq
Likewise, for the {\em{logarithmic mass terms}}
give rise to the following bilinear logarithmic terms (see~\cite[Lemma~4.6.2]{cfs})
\begin{align}
\frac{1}{2} &\, \Tr \big( i \slashed{\xi}\, \chi_L \, P(x,y) \big) \nonumber \\
&\asymp
\frac{m^2}{8} \Big( A_L[\xi]\: A_L[\xi]\, YY + 2 \,A_L[\xi]\: YY\: A_L[\xi]
+ YY \:A_L[\xi] \:A_L[\xi] \Big)\, T^{(1)} \label{l:mp2} \\
&\quad -\frac{m^2}{2}\: Y \:A_R[\xi] \:A_R[\xi]\: Y \:T^{(1)}\:. \label{l:mp3}
\end{align}
It is not obvious how the gauge phases are to be inserted into the contributions by the 
{\em{microlocal chiral transformation}}~\eqref{logmcl}. A naive computation yields a formula similar to~\eqref{log2j},
but with the chirality flipped. This gives rise to non-zero contributions to the EL equations already
of higher degree on the light cone (the so-called {\em{shear contributions}}; for details
see~\cite[\S3.7.11 and \S4.4.5]{cfs}). It turns out that, in order to make these contributions vanish,
one can modify the microlocal transformation by introducing the transformation~$U_\text{flip}$,
which involves new chiral potentials~$A^\text{even}_L$ and~$A^\text{even}_R$.
In the present situation of one generation, we simply choose~$A^\text{even}_L = A_R$
and~$A^\text{even}_R = A_L$. With this choice, the gauge phases are inserted into~\eqref{logmcl}
similar to~\eqref{log2j},
\beq \label{log2mlt}
\chi_L P(x,y) \asymp
-\frac{i}{2}\:\chi_L \left( A_L[\xi] \; \slashed{v}_L + \slashed{v}_L \;A_L[\xi] \right) \: T^{(1)} \:.
\eeq

Comparing the bilinear logarithmic terms~\eqref{log2j}--\eqref{log2mlt}, one sees that~\eqref{log2j}
and~\eqref{log2mlt} have a similar structure, whereas in the bilinear logarithmic mass terms~\eqref{l:mp2}
and~\eqref{l:mp3} the chiral potentials enter in a fundamentally different way.
This has the effect that the EL equations will be satisfied only if the following condition holds.
\begin{Lemma} \label{lemmalog}
The logarithmic bilinear contributions to the kernel of the fermionic projector drop out of the
EL equations if and only if the matrices~$B_L$ and~$B_R$ defined by
\begin{align*}
B_L :=\:& -\frac{1}{4} \:\Big( A_L[\xi]^2\, YY +2 A_L[\xi] \,YY A_L[\xi] + YY A_L[\xi]^2 \Big) \\
&+ A_L[\xi]\, Y A_R[\xi]\, Y - Y A_R[\xi]^2\, Y
+ Y A_R[\xi]\, Y A_L[\xi]
\end{align*}
agree and are a multiple of the matrix~$Y^2$.
\end{Lemma}
\Proof This statement of this lemma is proven in~\cite[Lemma~4.6.3]{cfs}, where~$B_L$ is
defined by (see~\cite[eq.~(4.6.21)]{cfs} with~$A_R^\text{even} = A_L$)
\begin{align*}
B_L :=\:& -\frac{m^2}{4} \Big\{ A_L[\xi], \left( A_L[\xi]\, YY - 2 Y A_R[\xi]\, Y + YY A_L[\xi]
\right) \Big\} \:T^{(1)} \\
&\quad\:+\frac{m^2}{8} \Big( A_L[\xi]^2  YY + 2 A_L[\xi] YY A_L[\xi]
+ YY A_L[\xi]^2 \xi_k) \Big)\, T^{(1)} \\
&\quad\: -\frac{m^2}{2}\: Y A_R[\xi]^2\: Y \:T^{(1)} \\
&=-\frac{m^2}{8} \Big( A_L[\xi]^2\, YY +2 A_L[\xi] \,YY A_L[\xi] + YY A_L[\xi]^2 \Big)\: T^{(1)} \\
&\quad\: +\frac{m^2}{2} \Big( A_L[\xi]\, Y A_R[\xi]\, Y - Y A_R[\xi]^2\, Y
+ Y A_R[\xi]\, Y A_L[\xi] \Big)\: T^{(1)}\:.
\end{align*}
Dividing by~$T^{(1)}$ and multiplying by~$2/m^2$ gives the result.
\QED

In general, it does not seem easy to evaluate what these conditions imply for the
dynamical gauge group. However, there is a simple special case.
\begin{Lemma} \label{lemmabillog}
Assume that the algebraic gauge potentials commute with the mass matrix, i.e.\
\beq \label{Ycommute}
A_L Y = Y A_L \qquad \text{and} \qquad A_R Y = Y A_R
\eeq
for all algebraic gauge potentials~$(A_L, A_R)$. Then the conditions of Lemma~\ref{lemmalog}
simplify to the condition that for any~$\xi$ on the light cone, there is a constant~$c=c(\xi)$ such that
\beq \label{quadratic}
\big( A_L[\xi] - A_R[\xi] \big)^2\: Y = c(\xi)\, Y \:.
\eeq
\end{Lemma}
\Proof Commuting the mass matrices to the left, we obtain
\begin{align*}
B_L &= -Y^2 \:\Big( A_L[\xi]^2 - A_L[\xi]\: A_R[\xi] + A_R[\xi]^2 - A_R[\xi]\: A_L[\xi] \Big) \\
&= -Y^2 \:\big( A_L[\xi] - A_R[\xi] \big)^2 \:.
\end{align*}
Evaluating the conditions of Lemma~\ref{lemmalog} for~$B_L$ and~$B_R$ of this form gives the result.
\QED

\subsection{Constraints from the Field Tensor Terms} \label{secfield}
Following~\cite[\S5.3.4]{cfs}, further constraints are obtained from the field tensor terms.
As in~\cite[Proposition~5.3.8]{cfs}, we obtain the following condition.
\begin{Lemma} \label{lemmafield}
Assume that the chiral symmetry is broken in the sense that one of
the matrices~$X_L$ or~$X_R$ are not equal to the identity. Then all admissible gauge potentials
satisfy the relation
\beq \label{tracecond}
\Tr \big( A_L[\xi]+A_R[\xi] \big) = 0 \:.
\eeq
\end{Lemma}

\begin{Def} \label{defadm} The {\bf{admissible gauge potentials}}~$(A_L, A_R)$ are formed of all
algebraic gauge potentials which also satisfy the constraints from the bilinear logarithmic terms
and the field tensor terms. They
generate the {\bf{admissible gauge group}}~$\G^\text{adm}$.
\end{Def}

\section{The Main Theorem} \label{secmain}
In the previous section, we compiled the
constraints for the gauge potentials~$(A_L, A_R)$ coming from the causal
action principle. We now explain how to get from these constraints to the
setting studied in our main theorem (Theorem~\ref{thmintro} in the introduction
or, equivalently, Theorem~\ref{thm:main}).
We first note that the constraints~\eqref{Aalg} and~\eqref{tracecond}
are linear in the gauge potentials. Therefore, they are easy to evaluate; they
can be implemented at any stage of the construction.
With this in mind, the main difficulty is to evaluate the constraint~\eqref{quadratic},
which is quadratic in the gauge potentials.
Our main theorem gives a systematic way to evaluate the these quadratic
constraints. In order for our theorem to apply, we need to assume that
the mass matrix~$Y$ is invertible and commutes with the chiral potentials~\eqref{Ycommute}.
Under these assumptions, Lemma~\ref{lemmabillog} applies, and the
constraint~\eqref{quadratic} can be stated that
\beq \label{cxi}
\big( A_L[\xi] - A_R[\xi] \big)^2 = c\: \1 \:,
\eeq
to be satisfied for all vectors~$\xi$ on the light cone
(here~$c$ may depend on~$\xi$ and~$(A_L, A_R)$).
In order to avoid direction-dependent constraints for the gauge potentials,
the condition~\eqref{cxi} must hold similarly for all gauge potentials, i.e.\
\beq \label{Acliff}
\big( A_L - A_R \big)^2 = c\: \1 \qquad \text{for all~$(A_L, A_R) \in \g^\text{adm}$}\:,
\eeq
where~$\g^\text{adm}$ denotes the Lie algebra of the gauge group~$\G^\text{adm}$
in Definition~\ref{defadm}. This is precisely the condition~\eqref{liecond}
stated in the introduction. Theorem~\ref{thmintro} gives a method for determining
all maximal Lie algebras which satisfy this condition.
Taking these Lie algebras and implementing the linear
constraints~\eqref{Aalg} and~\eqref{tracecond} gives candidates for the
Lie algebra of the admissible gauge group~$\G^\text{adm}$.

This general procedure will be illustrated in Section~\ref{secexamples}
in various examples. In this section, we proceed by formulating the problem
abstractly and giving a general proof.

\subsection{Description with Clifford Algebras}
Given a parameter~$N \in \N$, we work on a complex Hilbert space decomposed as
\[ H=H_L\oplus H_R \qquad \text{with} \qquad \dim \H_L = \dim \H_R = N\:. \]
Using a block matrix notation, we introduce the following operators on~$\H$,
\beq \label{GDS}
\Gamma=\begin{pmatrix}0&\mathbb{I}\\ \mathbb{I}&0\end{pmatrix},\qquad
D=\begin{pmatrix}0&i\mathbb{I}\\ -i\mathbb{I}&0\end{pmatrix},\qquad
S:=-i\,D\Gamma=\begin{pmatrix}\mathbb{I}&0\\ 0&-\mathbb{I}\end{pmatrix}.
\eeq
Then~$\Gamma^2=D^2=S^2=\mathbb{I}$ and~$\{S,\Gamma\}=0$, $\{D,\Gamma\}=0$, $\{S,D\}=0$. These three matrices generate a fixed copy of~$\mathbb{C}\ell(2)$ and provide a convenient basis. The operator~$S$ is the grading operator: it equals~$+I$ on~$H_L$ and~$-I$ on~$H_R$.  The operator~$\Gamma$ swaps the chirality: it interchanges~$H_L$ and~$H_R$.

In this description, chiral potentials are represented by a block-diagonal matrix
\beq \label{Achiral}
A = \begin{pmatrix} A_L & 0 \\[2pt] 0 & A_R \end{pmatrix} \:.
\eeq
We decompose~$A$ into~$\Gamma$–even (vectorial) and~$\Gamma$–odd (axial) parts via the adjoint action of~$\Gamma$,
\begin{align}
A_V &:= \frac{1}{2}\,\big(A+\Gamma A\Gamma\big) = \frac{1}{2}\begin{pmatrix} A_L{+}A_R&0\\[2pt]0&A_L{+}A_R\end{pmatrix}, \label{AV} \\[4pt]
A_A &:= \frac{1}{2}\,\big(A-\Gamma A\Gamma\big) = \frac{1}{2}\begin{pmatrix} A_L{-}A_R&0\\[2pt]0&-(A_L{-}A_R)\end{pmatrix}. \label{AA}
\end{align}

Polarizing the condition~\eqref{Acliff} (as explained before~\eqref{sproddef}
in the introduction), one sees that the admissible
matrices~$A_L - A_R$ must generate a Clifford algebra.
In order to formalize this condition, we work inside a complex even–dimensional Clifford algebra~$\bb{C}\ell(2n)$ acting on~$H$, with Hermitian generators~$\{\gamma_a\}_{a=1}^{2n}$ satisfying the anti-commutation relations
\[ 
\{\gamma_a,\gamma_b\} = 2\,\delta_{ab}\,\I \:. \]
We choose the Clifford algebra such that it contains the matrices~$\Gamma$
and~$S$ in~\eqref{GDS}.
For the admissible chiral potentials of the form~\eqref{Achiral},
the corresponding axial components~\eqref{AA} generate a subspace
of the grade one subspace of the Clifford algebra~$\bb{C}\ell(2n)$.
We denote this subspace by~$V \subset \bb{C}\ell^{[1]}(2n)$.
We fix an orthogonal decomposition of the grade one vector space
\[ \bb{C}\ell^{[1]}(2n) = \Span\{\gamma_a\}_{a=1}^{2n} = V \oplus W,
\qquad \dim V = p,\ \dim W = s,\ p+s=2n \:, \]
with bases~$\{\gamma_i\}_{i=1}^{p}\subset V$ and~$\{\gamma_\alpha\}_{\alpha=p+1}^{p+s}\subset W$. With this choice there is a graded-tensor identification
\[ \bb{C}\ell(V\oplus W) \cong \bb{C}\ell(V) \hat{\otimes} \bb{C}\ell(W)\:. \]

Finally, it is obvious from~\eqref{AA} that every~$v\in V$ is represented in
block form as~$v = \mathrm{diag}(\gamma, -\gamma)$.
As a consequence, the subspace~$V$ is orthogonal to both~$S$ and~$\Gamma$.
Moreover, every~$v \in \V$ is determined by its upper left block matrix entry.
We denote the vector space generated by these upper left block matrices by~$\mathscr{A}
\in \Symm(H_L)$.
Thus every~$v \in V$ can be written uniquely as
\[ v = \begin{pmatrix} \gamma & 0 \\ 0 & -\gamma \end{pmatrix}
\qquad \text{with~$\gamma \in \mathscr{A}$}\:. \]
Clearly, the vectors in~$\mathscr{A}$ also satisfy the anti-commutation relations
and generate a Clifford algebra~$\C\ell(\mathscr{A})$ represented on~$\H_L$.

\subsubsection{Grade Decomposition}

The Clifford algebra~$\bb{C}\ell(V)$ decomposes as a direct sum of grade-$r$ subspaces:
\[
  \bb{C}\ell(V) = \bigoplus_{r=0}^{p} \bb{C}\ell^{[r]}(V) \:, \]
where~$\bb{C}\ell^{[r]}(V)$ is spanned by all products~$\gamma_{i_1}\cdots\gamma_{i_r}$
with~$i_1 < \cdots < i_r$.  In particular:
\bitem
  \item $\bb{C}\ell^{[0]}(V) = \bb{C}\cdot I$ (scalars),
  \item $\bb{C}\ell^{[1]}(V) = V$ (1-vectors),
  \item $\bb{C}\ell^{[2]}(V) = \Span\{\gamma_i\gamma_j : i<j\}$ (bi-vectors),
  \item $\bb{C}\ell^{[p]}(V) = \bb{C}\cdot\omega$ (pseudoscalar, defined below).
\eitem
Every~$X\in\bb{C}\ell(V)$ is written uniquely as~$X = \sum_{r=0}^p \langle X \rangle_{r}$,
where~$\langle X\rangle_{r}\in\bb{C}\ell^{[r]}(V)$ is the grade-$r$ part of~$X$.

Since the Lie algebras in this paper are represented by Hermitian matrices
and carry the bracket~$i[\cdot,\cdot]$, we use the Hermitian real form of the
bi-vector space,
\[
\mathfrak{spin}(V):=\operatorname{span}_{\R}
\big\{i\gamma_j\gamma_k:1\leq j<k\leq p\big\} \:.
\]
It satisfies
\[
i[V,V]\subset\mathfrak{spin}(V),\qquad
i[\mathfrak{spin}(V),V]\subset V,\qquad
i[\mathfrak{spin}(V),\mathfrak{spin}(V)]\subset\mathfrak{spin}(V) \:.
\]
Consequently, $V\oplus\mathfrak{spin}(V)$ is closed under the bracket
$i[\cdot,\cdot]$. As a real Lie algebra, this sum is naturally isomorphic to
$\mathfrak{spin}(p+1)$. Indeed, if~$e_1,\ldots,e_{p+1}$ are the generators of a
Euclidean Clifford algebra in one higher dimension, the isomorphism sends
$\gamma_j$ to~$i e_j e_{p+1}$ and~$i\gamma_j\gamma_k$ to~$i e_j e_k$.
These are the only general facts about bi-vectors and spin algebras that will be
needed below.

We refer to the parity of the grade~$r$ as the Clifford grade parity. Thus
\[ \bb{C}\ell^{\mathrm{even}}(V):=\bigoplus_{\substack{0\leq r\leq p\\ r\ \mathrm{even}}}
\bb{C}\ell^{[r]}(V),\qquad
\bb{C}\ell^{\mathrm{odd}}(V):=\bigoplus_{\substack{0\leq r\leq p\\ r\ \mathrm{odd}}}
\bb{C}\ell^{[r]}(V) \:. \]

\subsubsection{The Pseudoscalar}

Define the pseudoscalar of~$V$ by
\[ 
  \omega := i^{\lfloor p/2\rfloor}\,\gamma_1\gamma_2\cdots\gamma_p \;\in\;\bb{C}\ell^{[p]}{V} \:. \]
The prefactor~$i^{\lfloor p/2\rfloor}$ is chosen to make~$\omega$ Hermitian.
The pseudoscalar satisfies the following relations,
\begin{align*}
  \omega^2 &= \bb{I} \\ 
  \omega\gamma_j &= (-1)^{p-1}\,\gamma_j\omega
    \quad\text{for all }j=1,\ldots,p \:. 
\end{align*}
In particular:
\bitem
  \item If~$p$ is even: $\omega\gamma_j = -\gamma_j\omega$
    (anti-commutes with every~$\gamma_j\in V$).
  \item If~$p$ is odd: $\omega\gamma_j = +\gamma_j\omega$
    (commutes with every~$\gamma_j\in V$, so~$\omega$ is central in~$\bb{C}\ell(V)$
    and lies in~$\Comm(V)$).
\eitem
Furthermore, $\omega$ commutes with every even element of~$\bb{C}\ell(V)$.

\subsubsection{The Commutant of~$V$}

\begin{Def} {\bf{(Commutant)}}
\label{def:commutant}
The commutant of~$V$ is defined as
\[ 
  \Comm(V) := \big\{ A \in \mathrm{End}(H) \:\big|\: [A,v] = 0\text{ for all }v\in V \big\} \:. \]
\end{Def}

The even subalgebra~$\bb{C}\ell^+(W)\subset\bb{C}\ell(W)$ is contained in~$\Comm(V)$: every even element of~$\bb{C}\ell(W)$ commutes with every element of~$\bb{C}\ell(V)$ by the graded tensor product rule stated above. However, the full commutant is generally larger than~$\bb{C}\ell^+(W)$. When~$p$ is odd, the pseudoscalar~$\omega\in\bb{C}\ell(V)$ commutes with every~$\gamma_j\in V$ and therefore lies in~$\Comm(V)$ as well.

\subsection{Admissibility Constraints and Statement of Main Theorem}
Let~$\mathfrak{g}\subset\mathrm{End}(H)$ be the gauge Lie algebra with Lie bracket~$i[\,\cdot\,,\,\cdot\,]$. We say that~$\mathfrak{g}$ is {\em{admissible}} if the following
two conditions hold:
\begin{align}
  \tag{A1}\label{eq:A1} &[S,\,Y] = 0 \quad\text{for all }Y\in\mathfrak{g},\\
  \tag{A2}\label{eq:A2} &Y^{(-)} \in V = \bb{C}\ell^{[1]}(V)
    \quad\text{for all }Y\in\mathfrak{g}.
\end{align}
Condition~\eqref{eq:A1} says every generator commutes with~$S$. Condition~\eqref{eq:A2} says the~$\Gamma$-odd part of every generator is a grade-1 element (a 1-vector) of~$\bb{C}\ell(V)$.

The results of the remaining sections establish the following classification theorem.
\begin{Thm} {\bf{(Main classification theorem)}}
\label{thm:main}
Let~$\mathfrak{g}\subset \text{End}(H)$ be an admissible gauge Lie algebra satisfying
the conditions
\beq \label{c1}
[S,Y]=0 \qquad \text{for all }Y\in\mathfrak{g}
\eeq
and
\beq \label{c2}
Y^{(-)}\in V=\bb{C}\ell^{[1]}(V)\qquad\text{for all }Y\in\mathfrak{g} \:.
\eeq
Then, depending on the value of the parameter~$p=\dim V$, the following
statements hold true.
\bitem
\item[{\rm{(a)}}] If~$p$ is odd, the generators of~$\g$ have the form
\[ g=\begin{pmatrix}
\gamma+\sigma+C & 0\\
0 & -\gamma+\sigma+C
\end{pmatrix} \]
with
\beq \label{gammsC}
\gamma\in \mathscr{A},\qquad\sigma\in\bb{C}\ell^{[2]}(\mathscr{A}),\qquad C\in\Comm(\mathscr{A}) \:.
\eeq
\item[{\rm{(b)}}] If~$p$ is even, there is a linear mapping~$B : \mathscr{A} \rightarrow \mathscr{A} \otimes \Comm(\mathscr{A})$ such that every generator~$g \in \g$ takes the form
\beq \label{gform}
g=\begin{pmatrix}
\big(\gamma +\omega B(\gamma) \big)+\sigma+C & 0\\
0 & \big(\gamma + \omega B(\gamma) \big) +\sigma+C
\end{pmatrix} ,
\eeq
again with~$\gamma$, $\sigma$ and~$C$ as in~\eqref{gammsC}.
Moreover, the linear mapping~$B$ has 
for all~$\gamma, \gamma' \in \mathscr{A}$ the following properties,
\begin{align}
\text{symmetry:} \quad \big\{ \gamma,\, B \gamma' \big\} &= \big\{ B \gamma,\, \gamma' \big\} 
\label{even1} \\
\big[ B \gamma, \sigma' \big] + \big[ \sigma, B \gamma' \big]
&= B \big( [\gamma, \sigma' ] + [ \sigma, \gamma' ] \big)  \\
\big[ B \gamma, B \gamma' \big] &\in \bb{C}\ell^{[2]}(\mathscr{A}) \oplus \Comm(\mathscr{A}) \:. \label{even3}
\end{align}
\eitem
\end{Thm} \noindent
We remark that a specific choice of~$B$ having the above
properties~\eqref{even1}--\eqref{even3}  is
\beq \label{Bsimpform}
B \gamma = b \gamma \qquad \text{with~$b \in \Comm(\mathscr{A})$ and~$b^2 \sim \1$}\:.
\eeq
However, we show in an example that~$B$ in general is not of this form
(see Example~\ref{excounter1}).
We conjecture that, if~$\g$ is semi-simple, then every linear mapping~$B$ should necessarily
of the form~\eqref{Bsimpform}. But, at present, we have no proof of this conjecture.

\subsection{The Block-Diagonal Constraint}
\label{sec:blockdiag}

The condition~\eqref{eq:A1} simply enforces block-diagonality and places no constraint whatsoever on Clifford grades. That is, for~$Y\in\mathrm{End}(H)$: $[S,Y]=0$ if and only if~$Y = \mathrm{diag}(Y_L, Y_R)$. The proof is straightforward.

\begin{Prp} {\bf{(Block-diagonality is preserved by the~$\Gamma$-decomposition)}}
\label{prop:A1preserved}
If~$[S,Y]=0$, then~$[S,Y^{(+)}]=0$ and~$[S,Y^{(-)}]=0$.
\end{Prp}

\begin{proof}
We show~$[S,\Gamma Y\Gamma]=0$ when~$[S,Y]=0$.  Using~$\{S,\Gamma\}=0$,
i.e.\ $S\Gamma = -\Gamma S$:
\begin{align*}
  [S,\Gamma Y\Gamma] &= S(\Gamma Y\Gamma) - (\Gamma Y\Gamma)S
  = -\Gamma(SY)\Gamma + \Gamma(YS)\Gamma
  = -\Gamma[S,Y]\Gamma = 0.
\end{align*}
Therefore, $[S,Y^{(+)}] = \frac{1}{2}([S,Y]+[S,\Gamma Y\Gamma]) = 0$, and
similarly~$[S,Y^{(-)}]=0$.
\end{proof}

From here on, all generators are block-diagonal (by condition~\eqref{eq:A1}). All further constraints come from condition~\eqref{eq:A2} and from Lie closure.

\subsection{The Grade-Shift Lemma}
\label{sec:gradeshift}

We now focus on a key algebraic result that is foundational to proving later results: commuting a 1-vector with a homogeneous element does not mix grades. A homogeneous Clifford element is an element lying entirely in one grade $\Cl^{[r]}(V)$, rather than a sum of components of different grades.

\subsubsection{Interior and Exterior Products}

\begin{Def} {\bf{(Interior and exterior products)}}
\label{def:intext}
For~$v\in V=\bb{C}\ell^{[1]}(V)$ and~$\alpha\in\bb{C}\ell^{[r]}(V)$, we define
\begin{align*}
  v\intprod\alpha &:= \frac{1}{2}(v\alpha - (-1)^r\alpha v)\in\bb{C}\ell^{[r-1]}(V) \\
  v\extprod\alpha &:= \frac{1}{2}(v\alpha + (-1)^r\alpha v)\in\bb{C}\ell^{[r+1]}(V) \:.
\end{align*}
These are the interior product (grade-lowering) and exterior product (grade-raising). Adding them gives the Clifford product: $v\alpha = (v\intprod\alpha) + (v\extprod\alpha)$.
\end{Def}

\begin{Lemma} {\bf{(Grade-Shift Lemma)}}
\label{lem:gradeshift}
Let~$v\in V$ and~$\alpha\in\bb{C}\ell^{[r]}(V)$.  Then
\[ 
  [v,\alpha] = \begin{cases}
    2(v\intprod\alpha)\in\bb{C}\ell^{[r-1]}(V) & \text{if }r\text{ is even},\\
    2(v\extprod\alpha)\in\bb{C}\ell^{[r+1]}(V) & \text{if }r\text{ is odd} \:.
  \end{cases} \]
In particular, the commutator of a 1-vector with a homogeneous grade-$r$ element is itself homogeneous: it lands in grade~$r-1$ (if~$r$ is even) or grade~$r+1$ (if~$r$ is odd).  No mixing of grades occurs.
\end{Lemma}

\begin{proof}
We compute directly from Definition~\ref{def:intext}.

If~$r$ is even, then~$(-1)^r=+1$, so
\[ v\intprod\alpha=\frac{1}{2}(v\alpha-\alpha v)=\frac{1}{2}[v,\alpha] \:. \]
Hence
\[ [v,\alpha]=2(v\intprod\alpha)\in\bb{C}\ell^{[r-1]}(V) \:. \]

If~$r$ is odd, then~$(-1)^r=-1$, so
\[ v\extprod\alpha=\frac{1}{2}(v\alpha-\alpha v)=\frac{1}{2}[v,\alpha] \:. \]
Hence
\[ [v,\alpha]=2(v\extprod\alpha)\in\bb{C}\ell^{[r+1]}(V) \:. \]

Thus the commutator of a~$1$-vector with a homogeneous grade-$r$ element is itself homogeneous: it lies in grade~$r-1$ when~$r$ is even, and in grade~$r+1$ when~$r$ is odd. In particular, no mixing of grades occurs.
\end{proof} \noindent
We see therefore that the commutator with a 1-vector reverses Clifford grade parity on each homogeneous component.

The commutator acts grade-wise on the homogeneous decomposition. For~$v\in V$ and
\[ X=\sum_{r=0}^p \langle X\rangle_r \in \bb{C}\ell(V) \:, \]
one has
\[ 
[v,X]=\sum_{\substack{0\leq r\leq p\\ \text{$r$ even}}}2\bigl(v\intprod \langle X\rangle_r\bigr)
+\sum_{\substack{0\leq r\leq p\\ \text{$r$ odd}}}2\bigl(v\extprod \langle X\rangle_r\bigr) \:. \]
Moreover, each summand on the right-hand side is homogeneous: for even~$r$ it lies in~$\bb{C}\ell^{[r-1]}(V)$, while for odd~$r$ it lies in~$\bb{C}\ell^{[r+1]}(V)$. Since
\[ \bb{C}\ell(V)=\bigoplus_{r=0}^p \bb{C}\ell^{[r]}(V) \]
is a direct sum, these homogeneous contributions are linearly independent. In particular, there is no cancellation between the contributions coming from different grades of~$X$.

\subsection{The Explicit Form of the Generators} \label{secgen}
Conditions~\eqref{eq:A1} and~\eqref{eq:A2} together imply that generators of~$\mathfrak{g}$ take one of two forms. We again let~$p=\dim V$ and~$q = \dim\mathfrak{g}$.
It will be convenient to work in a specific basis of~$\g$ of the following form.

\begin{Def} {\bf{(Type I and Type II generators)}}
\label{def:types}
The basis of~$\mathfrak{g}$ splits into:
\begin{align}
  g_j &= \begin{pmatrix}\gamma_j + A_j & 0 \\ 0 & -\gamma_j + A_j\end{pmatrix}
    && \text{{\rm{(Type I)}}},  &&j=1,\ldots,p,
    \label{eq:typeI}\\
  g_j &= \begin{pmatrix}A_j & 0 \\ 0 & A_j\end{pmatrix}
    && \text{{\rm{(Type II)}}}, && j=p+1,\ldots,q,
    \label{eq:typeII}
\end{align}
where~$\gamma_1,\ldots,\gamma_p\in V$ are orthonormal so that~$\{\gamma_j,\gamma_k\}=2\delta_{jk}$,
and the~$A_j\in\mathrm{End}(H)$ are symmetric matrices to be determined.
We refer to the~$\gamma_j$ and~$A_j$ as the {\bf{$\Gamma$-odd}} and {\bf{$\Gamma$-even
 potentials}}, respectively.
\end{Def} \noindent
Thus for the Type I generators, the~$\Gamma$-odd part is~${g_j}^{(-)} = \mathrm{diag}(\gamma_j,-\gamma_j)$. Therefore, $\gamma_j\in V\subset\bb{C}\ell^{[1]}(V)$, satisfying~\eqref{eq:A2}. The~$\Gamma$-even contribution~$A_j$ appears with the same sign in both diagonal blocks (making it~$\Gamma$-even), regardless of its Clifford grade.
For the Type II generators, on the other hand, ${g_j}^{(-)} = 0\in V$, so~\eqref{eq:A2} is trivially satisfied.

The fact that~$\mathfrak{g}$ is closed under the Lie bracket means that for all~$j,k \in \{1,\ldots q\}$, their commutator
\begin{align}
i \big[g_j, g_k \big] &= \begin{pmatrix} i \big[ \gamma_j, A_k \big] + i \big[ A_j, \gamma_k \big] & 0 \\ 0 &
-i \big[ \gamma_j, A_k \big] - i \big[ A_j, \gamma_k \big] \end{pmatrix} && \text{$\leftarrow$ $\Gamma$-odd} \label{codd} \\
&\quad\: + \begin{pmatrix} i \big[ \gamma_j, \gamma_k \big] + i
\big[ A_j, A_k \big] & 0 \\ 0
& i \big[ \gamma_j, \gamma_k \big] + i \big[ A_j, A_k \big] \end{pmatrix} && \text{$\leftarrow$ $\Gamma$-even} \label{ceven}
\end{align}
must be a linear combinations of the matrices in~\eqref{eq:typeI} and~\eqref{eq:typeII}.

Let us evaluate in detail what these conditions mean. We begin with the $\Gamma$-odd component~\eqref{codd}, which gives rise to the conditions
\[ i \big[ \gamma_j, A_k \big] + i \big[ A_j, \gamma_k \big] \in V \:, \]
to be satisfied for all~$j,k \in \{1,\ldots q\}$.
Since~$[A_j,\gamma_k] = -[\gamma_k,A_j]$, this can be written as 
\beq \label{cgen}
i[\gamma_j,A_k] - i[\gamma_k,A_j]\in V \qquad
\text{for all~$j,k \in \{1,\ldots q\}$}\:.
\eeq
Therefore, the difference of two commutators must be a 1-vector.

Clearly, these conditions are trivially satisfied if both~$j,k >p$. Therefore, we may assume that~$j \leq p$, and consider the two cases~$k \leq p$ and~$k>p$. 

In order to further simplify the analysis of~\eqref{cgen}, we note that, for any~$C \in \Comm(V)$ and~$\alpha\in\textrm{End}(H)$, the defining property of
$\Comm(V)$ gives
\[  [v,\,C\alpha] = C[v,\alpha]  \qquad \text{for all } v \in V \:. \]
Thus multiplication by an element of~$\Comm(V)$ does not change how commutators with~$V$ behave; it only contributes an overall commuting factor.
Accordingly, if
\[ A_k = \sum_\ell C_{k,\ell}\,\alpha_{k,\ell},
\qquad
C_{k,\ell}\in \Comm(V),\qquad \alpha_{k,\ell}\in \bb{C}\ell(V) \:, \]
then
\[ [\gamma_j,A_k]
= \sum_\ell C_{k,\ell}\,[\gamma_j,\alpha_{k,\ell}] \:. \]
This shows that the constraint~\eqref{cgen}
is determined entirely by the~$\bb{C}\ell(V)$-parts of the~$A_k$. Therefore, for the purpose of determining which Clifford grades are compatible with the vector-commutator constraints, it is sufficient to analyze the Clifford factors, and suppress the~$\Comm(V)$ factors. Once the admissible grades have been determined, the resulting elements may be multiplied by arbitrary elements of~$\Comm(V)$. Commutant factors do not change the resulting Clifford grade, but they must still be retained when checking the full Lie closure of the resulting generators.

\subsubsection{Coefficients of the Bi-vector Terms} \label{sec:bivectorcoefficients} The preceding discussion shows that commutant factors may be suppressed when determining the possible Clifford grades. They cannot, however, be chosen arbitrarily once the full constraint~\eqref{cgen} is imposed. We now consider the coefficients of the bi-vector terms. Decompose the representation of~$\bb{C}\ell(V)$ into irreducible components, grouping equivalent components together. On each resulting block, the Clifford algebra acts on the irreducible factor, whereas its commutant acts on the corresponding multiplicity factor. The non-scalar matrix components of a commutant coefficient can therefore be compared separately. Elements in the center of~$\bb{C}\ell(V)$ belong to the Clifford factor and are already covered by the grade analysis. Fix one non-scalar component of a commutant coefficient and denote the corresponding bi-vectors by~$X_k\in\bb{C}\ell^{[2]}(V)$. For a Type~II generator, the corresponding component of condition~\eqref{cgen} gives 
\[ \gamma_j\intprod X_k=0 \qquad\text{for all }j\in\{1,\ldots,p\} \:, \]
and hence~$X_k=0$. For the Type~I generators, the non-scalar component of~\eqref{cgen} gives 
\beq 
\gamma_j\intprod X_k=\gamma_k\intprod X_j \qquad\text{for all }j,k\in\{1,\ldots,p\}. \label{eq:bivectorcoeffsymmetry} 
\eeq 
Set 
\[ T:=\sum_{l=1}^p\gamma_l\extprod X_l \:. \]
Using~\eqref{eq:bivectorcoeffsymmetry} and the identity 
\[ \sum_{l=1}^p\gamma_l\extprod \big(\gamma_l\intprod X_k\big)=2X_k \:, \]
one obtains 
\begin{align*}
\gamma_k\intprod T &=X_k-\sum_{l=1}^p\gamma_l\extprod \big(\gamma_k\intprod X_l\big)\\ 
&=X_k-\sum_{l=1}^p\gamma_l\extprod \big(\gamma_l\intprod X_k\big)\\ 
&=-X_k \:. 
\end{align*}
Thus~$X_k=-\gamma_k\intprod T$. Since successive interior products anti-commute, it follows that 
\begin{align*}
\gamma_j\intprod X_k &=-\gamma_j\intprod\big(\gamma_k\intprod T\big)\\ 
&=\gamma_k\intprod\big(\gamma_j\intprod T\big)\\ 
&=-\gamma_k\intprod X_j \:. 
\end{align*} 
Together with~\eqref{eq:bivectorcoeffsymmetry}, this gives $\gamma_j\intprod X_k=0$ for all~$j,k$, and therefore~$X_k=0$. Consequently, a bi-vector term cannot carry a non-scalar coefficient on a multiplicity factor. Its scalar coefficient can be absorbed into the bi-vector itself, so that the surviving bi-vector contribution lies in~$\bb{C}\ell^{[2]}(V)$. This conclusion is specific to the bi-vector terms. For the grade-$(p-1)$ twist considered in Section~\ref{sectwist}, condition~\eqref{cgen} requires the coefficient matrix to be symmetric in its Clifford vector indices, but does not require its entries in~$\Comm(V)$ to be scalar.

\subsection{Grade Restrictions for the $\Gamma$-Even Potentials}
\label{sec:graderestrict}
We shall now apply condition~\eqref{cgen} systematically to determine which Clifford grades can appear in the potentials~$A_k$.

\subsubsection{Type II Generators}

For~$j\leq p$ and~$k>p$ (Type II), equation~\eqref{cgen} reduces to the individual condition
\begin{equation}
\label{eq:individual}
  [\gamma_j, A_k] \;\in\; V \quad\text{for all }j=1,\ldots,p.
\end{equation}
This is the condition that commuting~$A_k$ with any 1-vector in~$V$ produces a 1-vector.  We can now apply the Grade-Shift Lemma directly to classify which grades in~$A_k$ are compatible with this.

\begin{Corollary} {\bf{(Type II grade restriction)}}
\label{cor:typeII}
For~$k\in\{p+1,\ldots,q\}$, the $\Gamma$-even potential must lie in
\[ A_k \;\in\; \bb{C}\ell^{[2]}(V) \;\oplus\; \Comm(V) \:. \]
\end{Corollary}

\begin{proof}
We go through the grades of~$A_k$ within~$\bb{C}\ell(V)$. By the direct-sum decomposition of~$\bb{C}\ell(V)$ into homogeneous grade subspaces, the different grade contributions can be treated independently.
\begin{itemize}[leftmargin=1.5em]
\item \textit{Grade~$r=0$ (scalar):} $[\gamma_j, \lambda I] = 0 \in V$.  Always admissible.
\item \textit{Grade~$r=1$:} By the Grade-Shift Lemma (odd~$r$):
\[ [\gamma_j, A_k^{[1]}] = 2(\gamma_j \extprod A_k^{[1]}) \;\in\; \bb{C}\ell^{[2]}(V) \:. \]
For this to lie in~$\bb{C}\ell^{[1]}(V)=V$, it must be zero. Hence~$A_k^{[1]}=0$.
\item \textit{Grade~$r=2$ (bi-vector):} By the Grade-Shift Lemma (even~$r$):
\[ [\gamma_j, A_k^{[2]}] = 2(\gamma_j\intprod A_k^{[2]}) \;\in\; \bb{C}\ell^{[1]}(V) = V \:. \]
The output automatically lies in~$V$.  bi-vectors are always admissible.
\item \textit{Grade~$r$ even, $r\geq 4$:} By the Grade-Shift Lemma (even~$r$):
\[ [\gamma_j, A_k^{[r]}] = 2(\gamma_j\intprod A_k^{[r]}) \;\in\; \bb{C}\ell^{[r-1]}(V) \:. \]
Since~$r-1\geq 3$, this is not in~$V$ unless zero. Hence~$A_k^{[r]}=0$.
\item \textit{Grade~$r$ odd, $r\geq 3$:} By the Grade-Shift Lemma (odd~$r$):
\[ [\gamma_j, A_k^{[r]}] = 2(\gamma_j\extprod A_k^{[r]}) \;\in\; \bb{C}\ell^{[r+1]}(V) \:. \]
Since~$r+1\geq 4$, this is not in~$V$ unless zero.  Hence~$A_k^{[r]}=0$.
\item \textit{Grade~$r=p$ ($p$ odd):} $\omega$ is
central in~$\bb{C}\ell(V)$ for~$p$ odd, so~$[\gamma_j,\omega]=0\in V$. The pseudoscalar is admissible, but~$\omega\in\Comm(V)$ for~$p$ odd, so it is already included in the~$\Comm(V)$ factor.
\end{itemize}
Combining the above results, within~$\bb{C}\ell(V)$, only grade 0 and grade 2 survive. Including~$\Comm(V)$ gives~$A_k\in\bb{C}\ell^{[2]}(V)\oplus\Comm(V)$.
\end{proof}

\subsubsection{Type I Generators}
For~$j,k\leq p$ (both Type I), we need to take into account both summands in~\eqref{cgen}.
This is a much weaker condition than~\eqref{eq:individual}. Again, we analyze it grade by grade. Since
\[ [\gamma_j,A_k]-[\gamma_k,A_j]\in V=\bb{C}\ell^{[1]}(V) \:, \]
its homogeneous component of grade~$a$ must vanish for every~$a\neq 1$:
\beq \label{gradea}
\big\langle[\gamma_j,A_k]-[\gamma_k,A_j]\big\rangle_a=0,
  \quad\text{for all }a\neq 1.
\eeq
Write
\[ A_k=\sum_{r=0}^p A_k^{[r]} \:. \]
By Lemma~\ref{lem:gradeshift}, commutation with a $1$-vector reverses grade parity on each homogeneous component~$A_k^{[r]}$: even grade~$A_k^{[r]}$ contribute only to odd output grades, while odd grade~$A_k^{[r]}$ contribute only to even output grades. Thus a contribution coming from~$A_k^{[r]}$ with~$r$ even can never cancel one coming from an~$A_j^{[l]}$ with~$l$ odd. Moreover, for a fixed output grade~$a$, there is exactly one possible source grade: namely~$r=a+1$ if~$a$ is odd, and~$r=a-1$ if~$a$ is even. Therefore, evaluating~\eqref{gradea} in grade~$a$ compares only equal grades of both~$A_k$ and~$A_j$.

\begin{Lemma} {\bf{(Even grades~$r\geq 4$ vanish)}}
\label{lem:even4}
For all~$k\in\{1,\ldots,p\}$ and even~$r\geq 4$,
\[ A_k^{[r]}=0 \:. \]
\end{Lemma}
\begin{proof}
Fix an even~$r\geq 4$, and write
\[ A_k:=A_k^{[r]}\in \bb{C}\ell^{[r]}(V) \:. \]
By the Grade-Shift Lemma,
\[ [\gamma_j,A_k]=2 \:\big(\gamma_j\intprod A_k \big) \in \bb{C}\ell^{[r-1]}(V) \:. \]
Since~$r-1\geq 3>1$, condition~\eqref{gradea} implies
\begin{equation}
\label{eq:symcond_even}
\gamma_j\intprod A_k=\gamma_k\intprod A_j
\qquad
\text{for all }j,k\in\{1,\ldots,p\}.
\end{equation}

Now define
\[ T:=\sum_{\ell=1}^p \gamma_\ell\extprod A_\ell \;\in\; \bb{C}\ell^{[r+1]}(V) \:. \]
For any multi-vector~$X\in\bb{C}\ell(V)$, one has
\[ \gamma_k\intprod(\gamma_\ell\extprod X)=\delta_{k\ell}\,X-
\gamma_\ell\extprod(\gamma_k\intprod X) \:. \]
Therefore,
\begin{align*}
\gamma_k\intprod T
&=
\sum_{\ell=1}^p\gamma_k\intprod(\gamma_\ell\extprod A_\ell)
= \sum_{\ell=1}^p
\Bigl(\delta_{k\ell}A_\ell-\gamma_\ell\extprod(\gamma_k\intprod A_\ell)\Bigr)
=
A_k-\sum_{\ell=1}^p \gamma_\ell\extprod(\gamma_k\intprod A_\ell).
\end{align*}
By~\eqref{eq:symcond_even},
\[ \gamma_k\intprod A_\ell=\gamma_\ell\intprod A_k \:, \]
so
\[ \gamma_k\intprod T=A_k-\sum_{\ell=1}^p \gamma_\ell\extprod(\gamma_\ell\intprod A_k) \:. \]

For any homogeneous~$X\in\bb{C}\ell^{[r]}(V)$,
\[ \sum_{\ell=1}^p \gamma_\ell\extprod(\gamma_\ell\intprod X)=rX \:, \]
since each grade~$r$ basis element is recovered once for each of its~$r$ indices.
Applying this to~$X=A_k$ gives
\[ \sum_{\ell=1}^p \gamma_\ell\extprod(\gamma_\ell\intprod A_k)=rA_k \:. \]
Hence
\[ \gamma_k\intprod T = A_k-rA_k = (1-r) \:A_k \:, \]
and therefore
\begin{equation}
\label{eq:Ak_from_T_even}
A_k=\frac{1}{1-r}\,\gamma_k\intprod T.
\end{equation}
Thus every~$A_k$ is obtained by contracting the same~$(r+1)$-vector~$T$.

Applying~$\gamma_j\intprod$ to~\eqref{eq:Ak_from_T_even}, we obtain
\[ \gamma_j\intprod A_k=\frac{1}{1-r}\,(\gamma_j\intprod)\gamma_k\intprod T \:. \]
But interior products anti-commute, so
\[ (\gamma_j\intprod)\gamma_k\intprod T=-(\gamma_k\intprod)\gamma_j\intprod T \:. \]
Using~\eqref{eq:Ak_from_T_even} again, this gives
\[ \gamma_j\intprod A_k=-\gamma_k\intprod A_j \:. \]
Comparing with the symmetry condition~\eqref{eq:symcond_even}, we conclude that
\[ \gamma_j\intprod A_k=0\qquad\text{for all }j,k \:. \]

Finally, fix~$k$. Since~$\gamma_\ell\intprod A_k=0$ for every~$\ell\in\{1,\ldots,p\}$, we have
\[ \sum_{\ell=1}^p \gamma_\ell\extprod(\gamma_\ell\intprod A_k)=0 \:. \]
But~$A_k\in\bb{C}\ell^{[r]}(V)$, so that
\[ \sum_{\ell=1}^p \gamma_\ell\extprod(\gamma_\ell\intprod A_k)=rA_k \:. \]
It follows that~$rA_k=0$. Since~$r\geq 4$, we conclude that~$A_k=0$ for all~$k$.
\end{proof}

\begin{Lemma} {\bf{(Odd grades~$1\leq r\leq p-3$ vanish)}}
\label{lem:odd}
For all~$k\in\{1,\ldots,p\}$ and all odd~$r$ with~$1\leq r\leq p-3$,
\[ A_k^{[r]}=0 \:. \]
\end{Lemma}

\begin{proof}
Fix an odd~$r$ with~$1\leq r\leq p-3$, and write
\[ A_k:=A_k^{[r]}\in \bb{C}\ell^{[r]}(V) \:. \]
By the Grade-Shift Lemma,
\[ [\gamma_j,A_k]=2 \: \big(\gamma_j\extprod A_k \big) \in \bb{C}\ell^{[r+1]}(V) \:. \]
Since~$r+1\geq 2$ and~$r+1\leq p-2$, condition~\eqref{gradea} implies
\begin{equation}
\label{eq:symcond_odd}
\gamma_j\extprod A_k=\gamma_k\extprod A_j
\qquad
\text{for all }j,k\in\{1,\ldots,p\}.
\end{equation}

Now define
\[ T:=\sum_{\ell=1}^p \gamma_\ell\intprod A_\ell
\;\in\; \bb{C}\ell^{[r-1]}(V) \:. \]
For any multi-vector~$X\in\bb{C}\ell(V)$, one has the identity
\[ \gamma_k\extprod(\gamma_\ell\intprod X)=\delta_{k\ell}\,X-\gamma_\ell\intprod(\gamma_k\extprod X) \:. \]
Therefore,
\begin{align*}
\gamma_k\extprod T
&= \sum_{\ell=1}^p\gamma_k\extprod(\gamma_\ell\intprod A_\ell)
= \sum_{\ell=1}^p\Bigl(\delta_{k\ell}A_\ell-\gamma_\ell\intprod(\gamma_k\extprod A_\ell)\Bigr) \\
&= A_k-\sum_{\ell=1}^p \gamma_\ell\intprod(\gamma_k\extprod A_\ell) \:.
\end{align*}
By~\eqref{eq:symcond_odd},
\[ \gamma_k\extprod A_\ell=\gamma_\ell\extprod A_k \:, \]
so that
\[ \gamma_k\extprod T=A_k-\sum_{\ell=1}^p \gamma_\ell\intprod(\gamma_\ell\extprod A_k) \:. \]

For any homogeneous~$X\in\bb{C}\ell^{[r]}(V)$,
\[ \sum_{\ell=1}^p \gamma_\ell\intprod(\gamma_\ell\extprod X)=(p-r)X \:, \]
since each grade~$r$ basis element is recovered once for each of the~$p-r$ basis vectors which are not already present in it. Applying this to~$X=A_k$ gives
\[ \sum_{\ell=1}^p \gamma_\ell\intprod(\gamma_\ell\extprod A_k)=(p-r)A_k \:. \]
Hence
\[ \gamma_k\extprod T = A_k-(p-r)A_k = -(p-r-1)A_k \:, \]
and therefore
\begin{equation}
\label{eq:Ak_from_T_odd}
A_k=-\frac{1}{p-r-1}\: \gamma_k\extprod T.
\end{equation}
This is well-defined because~$r\leq p-3$, so~$p-r-1\geq 2$. Thus every~$A_k$ is obtained by exterior multiplication of the same $(r-1)$-vector~$T$.

Applying~$\gamma_j\extprod$ to~\eqref{eq:Ak_from_T_odd}, we obtain
\[ \gamma_j\extprod A_k=-\frac{1}{p-r-1}\,\gamma_j\extprod\gamma_k\extprod T \:. \]
But exterior products anti-commute, so
\[ \gamma_j\extprod\gamma_k\extprod T=-\gamma_k\extprod\gamma_j\extprod T \:. \]
Using~\eqref{eq:Ak_from_T_odd} again, this gives
\[ \gamma_j\extprod A_k=-\gamma_k\extprod A_j \:. \]
Comparing with the symmetry condition~\eqref{eq:symcond_odd}, we conclude that
\[ \gamma_j\extprod A_k=0\qquad\text{for all }j,k \:. \]

Finally, fix~$k$. Since~$\gamma_\ell\extprod A_k=0$ for every~$\ell\in\{1,\ldots,p\}$, we have
\[ \sum_{\ell=1}^p \gamma_\ell\intprod(\gamma_\ell\extprod A_k)=0 \:. \]
But~$A_k\in\bb{C}\ell^{[r]}(V)$, so that
\[ \sum_{\ell=1}^p \gamma_\ell\intprod(\gamma_\ell\extprod A_k)=(p-r)A_k \:. \]
It follows that~$(p-r)A_k=0$.
Since~$r\leq p-3$, we have~$p-r\geq 3$, and hence~$A_k=0$ for all~$k$.
\end{proof}
\begin{Remark} {\bf{(The case~$r=p-1$ must be treated separately)}} \label{rempm1}
{\em{
The previous lemma applies only for odd grades~$r$ with~$1\leq r\leq p-3$. For~$r=p-1$ (which is odd only when~$p$ is even), the argument breaks down. Indeed,
\[ [\gamma_j,A_k^{[p-1]}]=2 \:(\gamma_j\extprod A_k^{[p-1]})\in\bb{C}\ell^{[p]}(V)=\bb{C}\,\omega \:, \]
so the commutator takes values in the one-dimensional top-grade space. Moreover, in the proof of Lemma~\ref{lem:odd}, the reconstruction step
\[ A_k=-\frac{1}{p-r-1}\,\gamma_k\extprod T \]
is no longer available, because~$p-r-1=0$ when~$r=p-1$. Thus the previous symmetry-versus-antisymmetry argument gives no contradiction in this case. The grade~$r=p-1$ must therefore be analyzed separately. }} \QEDrem
\end{Remark}

\begin{Prp} {\bf{(Summary of possible grades for the~$\Gamma$-even potentials)}}
\label{prop:remaininggrades}
For Type II generators~$k\in\{p+1,\ldots,q\}$, one has
\[ A_k\in \bb{C}\ell^{[2]}(V)\oplus\Comm(V) \:. \]

For Type I generators~$k\in\{1,\ldots,p\}$, one has
\beq \label{type1cases}
A_k\in\begin{cases}
\bb{C}\ell^{[2]}(V)\oplus\Comm(V) & \text{if }p\text{ is odd},\\
\bb{C}\ell^{[2]}(V)\oplus\bigl(\Comm(V)\otimes \bb{C}\ell^{[p-1]}(V)\bigr)\oplus\Comm(V)
& \text{if }p\text{ is even} \:.
\end{cases}
\eeq
\end{Prp}

\begin{proof}
The Type II statement is exactly Corollary~\ref{cor:typeII}.

For Type I generators, Lemma~\ref{lem:even4} rules out all even grades~$r\geq 4$, while Lemma~\ref{lem:odd} eliminates all odd grades~$1\leq r\leq p-3$. Grade~$r=2$ always survives, since the Grade-Shift Lemma gives
\[ [\gamma_j,A_k^{[2]}]\in V \]
automatically. The scalar part is included in~$\Comm(V)$.

If~$p$ is odd, then the pseudoscalar satisfies~$\omega\in\Comm(V)$, so no further independent term survives beyond the bi-vector and commutant parts. If~$p$ is even, the only additional remaining possibility is the grade-$(p-1)$ term singled out in Remark~\ref{rempm1}. This gives exactly the stated form.
\end{proof}

\subsection{Symmetry of the Twist Term} \label{sectwist}
In the case that~$p$ is odd, Proposition~\ref{prop:remaininggrades}
immediately gives part~(a) of Theorem~\ref{thm:main}.
Therefore, it remains to consider the case that~$p$ is even.
In this case, we need to specify the form of the Type I generators.
More precisely, we our goal is to show that the contribution to~$A_k$
of grade~$p-1$, which according to~\eqref{type1cases} lies in the space
\[ 
A_k^{[p-1]}\in \Comm(V)\otimes\bb{C}\ell^{[p-1]}(V) \:. \]
We refer to this potential as the {\em{twist term}}.
We note that
the Clifford grade subspace~$\bb{C}\ell^{[p-1]}(V)$ is a $p$-dimensional vector space, and a convenient basis is
\[ \{\omega\gamma_1,\ldots,\omega\gamma_p\} \:. \]
Therefore, each grade-$(p-1)$ term can be written uniquely as
\beq \label{eq:Akpminus1}
A_k^{[p-1]}=\sum_{l=1}^p B_k^{l}\:\omega\gamma_l\qquad\text{with}\qquad B^k_{l}\in\Comm(V) \:. \eeq
We refer to~$B_k^{l}$ as the {\em{twist matrix}}.

Since we are in the Type I case, the anti-symmetrized condition~\eqref{gradea} to the output grade~$a=p$, gives
\begin{equation}
\label{eq:toppgradecond}
\big\langle[\gamma_j,A_k^{[p-1]}]\big\rangle_p=\big\langle[\gamma_k,A_j^{[p-1]}]\big\rangle_p
\qquad \text{for all }j,k\in\{1,\ldots,p\} \:.
\end{equation}
Substituting~\eqref{eq:Akpminus1}, and using that each~$B_k^{l}\in\Comm(V)$ commutes with every~$\gamma_j$, we find
\[ \Big[ \gamma_j,A_k^{[p-1]} \Big]=\sum_{l=1}^p B_k^{l}\,[\gamma_j,\omega\gamma_l] \:. \]
Since~$p$ is even, $\omega\gamma_j=-\gamma_j\omega$ for all~$j$. Hence
\begin{align*}
[\gamma_j,\omega\gamma_l] &=\gamma_j\omega\gamma_l-\omega\gamma_l\gamma_j
= -\omega\gamma_j\gamma_l-\omega\gamma_l\gamma_j \\
&= -\omega\{\gamma_j,\gamma_l\} = -2\:\delta_{jl} \:\omega.
\end{align*}
Therefore,
\[ 
[\gamma_j,A_k^{[p-1]}]=-2\: B_k^{j} \,\omega \:. \]
This is a Clifford factor of top grade, namely~$\omega\in\bb{C}\ell^{[p]}(V)$, so taking the grade-$p$ projection does nothing,
\[ \big\langle[\gamma_j,A_k^{[p-1]}]\big\rangle_p=-2\:B_k^{j}\,\omega \:. \]
Similarly,
\[ \big\langle[\gamma_k,A_j^{[p-1]}]\big\rangle_p=-2\:B_j^{k}\,\omega \:. \]
Comparing the two sides of~\eqref{eq:toppgradecond}, we conclude that
\beq \label{eq:Bsymm}
B_k^{j}=B_j^{k}\qquad \text{for all }j,k\in\{1,\ldots,p\} \:.
\eeq
Thus the coefficient matrix~$(B_k^{j})$ is symmetric in the Clifford vector indices.

We conclude with an example which in which the twist term is {\em{not}} of the 
form~\eqref{Bsimpform}
\begin{Example} \label{excounter1} {\em{
We consider the case~$p=2$ and choose~$\mathscr{A}$ as the Clifford algebra
generated by the Pauli matrices~$\gamma_1=\sigma_1$ and~$\gamma_2=\sigma_2$.
Then the pseudoscalar matrix is~$\omega = i \sigma_1 \sigma_2 = - \sigma_3$.
We choose
\[ g_j = \begin{pmatrix} \gamma_j + A_j & 0 \\ 0 & -\gamma_j + A_j \end{pmatrix} \]
with~$j=1,2$ and
\beq \label{A12}
A_1 = -i \omega \gamma_2 = \gamma_1 \qquad \text{and} \qquad
A_2 = -i \omega \gamma_1 = \gamma_2 \:.
\eeq
Consequently,
\[ g_1 = \begin{pmatrix} 2 \gamma_1 & 0 \\ 0 & 0 \end{pmatrix} \qquad \text{and} \qquad
g_2 = \begin{pmatrix} 0 & 0 \\ 0 & -2 \gamma_2 \end{pmatrix} \:. \]
These potentials generate a trivial two-dimensional Lie algebra~$\g$.
Moreover, the potentials~$A_j$ clearly satisfy the conditions~\eqref{c1} and~\eqref{c2}.
But, comparing~\eqref{A12} with~\eqref{gform}, one sees that
\[ B \gamma_1 = -i \gamma_2 \qquad \text{and} \qquad B \gamma_2 = -i \gamma_1 \:, \]
showing that the twist operator is not of the form~\eqref{Bsimpform}.
}}
\QEDrem
\end{Example}

\subsection{Completing the Proof}
\Proof[Proof of Theorem~\ref{thm:main}.]
The theorem now follows by assembling the preceding results. Section~\ref{sec:blockdiag} gives block-diagonality, Section~\ref{secgen} gives an explicit form of the generators of~$\g$.
Section~\ref{sec:graderestrict} determines the admissible Clifford grades,
Section~\ref{sectwist} establishes that the twist term is symmetric~\eqref{eq:Bsymm}.
This completes the proof.
\QED

\section{Examples} \label{secexamples}
\subsection{The General Procedure}
\label{sec:general-procedure}
The goal of this section is to illustrate in simple examples how
our main theorem (in the formulas of Theorem~\ref{thmintro} or
Theorem~\ref{thm:main}) can be applied.
Physically interesting examples make it necessary to choose~$N=8$
and include a mixing of the generations; this will be studied systematically
in a forthcoming publication.

Before entering the analysis of specific systems, we explain the general procedure.
We work at the level of the gauge Lie algebra.  In the simplified setting with~$N$ sectors, the
chiral gauge group is
\[ \G=\U(N)_L\times \U(N)_R \:, \]
and the corresponding gauge potentials are pairs
\[ (A_L,A_R)\in \mathfrak{u}(N)\oplus\mathfrak{u}(N) \:. \]
Elements of~$\mathfrak{u}(N)$ are represented
by Hermitian matrices, and the Lie bracket is given by~$i[\cdot,\cdot]$.

The basic question is which choices of maximal Lie subalgebras~$\g^\text{adm}
\subset \mathfrak{u}(N)\oplus\mathfrak{u}(N)$ satisfy the Clifford compatibility 
conditions~\eqref{liecond}. Here our main theorem gives a systematic procedure
for answering this question: We choose a Clifford algebra~$\mathscr{A}$ of dimension~$p$
represented on~$\C^N$. For simplicity, we only discuss the case that~$p$ is even
(in the case that~$p$ is odd, the argument is similar, but one must also choose take into
account the twist operator~$B$).
Then the Lie algebra~$\g$ must be spanned by chiral potentials of the form
\begin{align*}
A_L &= \gamma+\sigma+C\:, \\
A_R &= -\gamma+\sigma+C\:,
\end{align*}
with
\[ \gamma\in V,\qquad \sigma\in \bb{C}\ell^{[2]}(\mathscr{A}),\qquad C\in\Comm(\mathscr{A}) \:. \]
In particular,
\[ \frac{1}{2}(A_L-A_R)=\gamma,\qquad 
\frac{1}{2}(A_L+A_R)=\sigma+C \:. \]
Consequently, as vector spaces of possible axial and vectorial parts,
\[ \g_A= \mathscr{A},\qquad \g_V= \bb{C}\ell^{[2]}(\mathscr{A})
\oplus \Comm(\mathscr{A}) \:. \]
For clarity, we point out that
the decomposition into vectorial and axial parts is not a direct sum of Lie
algebras, because Clifford one-vectors do not close under commutation. Lie
closure forces the accompanying bi-vector directions to be included.

The admissible gauge algebra is required to be maximal with respect to the
Clifford compatibility condition.  Thus a Clifford subspace~$V'$ which can be enlarged to a larger Clifford subspace~$V\subset\Symm(\C^N)$, with the corresponding Lie algebra still satisfying the Clifford compatibility condition, gives only an admissible subgroup of a larger admissible group.  Such a subspace should not be listed as a final admissible gauge group.

\subsection{The Case~$N=2$}
\label{sec:N2}

We must classify maximal Clifford subspaces
\[ V\subset \Symm(\C^2) \:. \]
The Hermitian $2\times2$-matrices decompose as
\[ 
\Symm(\C^2)=\R\,\1_2\oplus
\Span_{\R}\{\sigma_1,\sigma_2,\sigma_3\} \:. \]
where~$\sigma_1,\sigma_2,\sigma_3$ are the Pauli matrices. Any Clifford subspace of~$\Symm(\C^2)$ is either contained in the scalar subspace~$\R\,\1_2$ or else in the traceless Hermitian subspace~$\Span_{\R}\{\sigma_1,\sigma_2,\sigma_3\}$.

If~$V$ is contained in the scalar subspace, then maximality forces~$V$ to be 
\[ 
V=\R\,\1_2 \:. \]

The traceless Hermitian subspace is the real generating space for the standard Pauli representation of the complex Clifford algebra~$\bb{C}\ell(3)$:
\[ 
\frac{1}{2}\{\sigma_i,\sigma_j\}=\delta_{ij}\,\1_2 \:. \]
Thus, in the non-scalar case, maximality gives
\[ 
V=\Span_{\R}\{\sigma_1,\sigma_2,\sigma_3\} \:. \]
Proper non-zero subspaces, such as~$\R\sigma_3$ or~$\Span_{\R}\{\sigma_1,\sigma_2\}$, give admissible subalgebras but are not maximal.

\subsubsection{The Central Case}
\label{sec:N2scalar}

First take
\[ \mathscr{A}=\R\,\1_2 \:. \]
This is a one-dimensional Clifford space, so it generates~$\bb{C}\ell(1)$. Since there
is only one Clifford generator,
\[ \bb{C}\ell^{[2]}(\mathscr{A})=0 \:. \]
Moreover,
\[ \Comm(\mathscr{A})=\mathfrak{u}(2) \:. \]
Since~$p=1$ is odd, the space of axial potentials is
\[ \g_A=\R\,\1_2\simeq\mathfrak{u}(1)_A \:, \]
whereas the space of vectorial potentials is
\[ \g_V=\mathfrak{u}(2)_V \:. \]
Therefore
\[ 
\g^{\text{adm}}=\mathfrak{u}(2)_V\oplus\mathfrak{u}(1)_A \:. \]
Splitting
\[ \mathfrak{u}(2)_V=\mathfrak{su}(2)_V\oplus\mathfrak{u}(1)_V \]
and subsequently recombining,
\[ \mathfrak{u}(1)_V\oplus\mathfrak{u}(1)_A\simeq\mathfrak{u}(1)_L\oplus\mathfrak{u}(1)_R \:, \]
we obtain
\[ 
\g^{\text{adm}}=\mathfrak{su}(2)_V\oplus\mathfrak{u}(1)_L\oplus\mathfrak{u}(1)_R \:. \]
Thus, the admissible gauge group is,
\[ 
\G^{\text{adm}}=\SU(2)_V\times\U(1)_L\times\U(1)_R \:. \]

\subsubsection{The Full Pauli Case}
\label{sec:N2Pauli}

The maximal non-central Clifford subspace is
\[ \mathscr{A}=\Span_{\R}\{\sigma_1,\sigma_2,\sigma_3\} \:, \]
which generates~$\bb{C}\ell(3)$. Since~$p=3$ is odd, we set~$B=0$. The Hermitian real form
of the bi-vectors is the Pauli space,
\[ \bb{C}\ell^{[2]}(\mathscr{A})_{\text{Herm}}
=\Span_{\R}\{i\sigma_1\sigma_2,i\sigma_2\sigma_3,i\sigma_3\sigma_1\}
=\Span_{\R}\{\sigma_1,\sigma_2,\sigma_3\} \:, \]
up to signs. Moreover, within~$\Symm(\C^2)$,
\[ \Comm(\mathscr{A})=\R\,\1_2 \:. \]
Therefore
\[ \g_A=\mathscr{A}\simeq\mathfrak{su}(2)_A,
\qquad
\g_V\simeq\mathfrak{su}(2)_V\oplus\mathfrak{u}(1)_V \:. \]
The vectorial and axial Pauli directions can equivalently be rewritten as independent left- and right-handed generators.  Let~$t_i$ denote a basis element of the Pauli~$\mathfrak{su}(2)$, and write an element of~$\mathfrak{u}(2)_L\oplus\mathfrak{u}(2)_R$ as a pair~$(A_L,A_R)$.  Then the
vectorial and axial generators are respectively
\[ (t_i,t_i) \qquad\text{and}\qquad (t_i,-t_i) \:. \]
Taking linear combinations gives
\begin{align*}
(t_i,0) &= \frac{1}{2}\big((t_i,t_i)+(t_i,-t_i)\big) \\
(0,t_i) &= \frac{1}{2}\big((t_i,t_i)-(t_i,-t_i)\big) \:.
\end{align*}
Thus the same traceless generators may be described either in terms of vectorial
and axial Pauli directions, or in terms of independent left- and right-handed
Pauli directions.  In the latter basis, the Lie algebra is
\[ \mathfrak{su}(2)_L\oplus\mathfrak{su}(2)_R \:. \]
The scalar commutant remains vectorial. We therefore obtain
\[ 
\g^{\text{adm}}=\mathfrak{u}(1)_V\oplus\mathfrak{su}(2)_L\oplus\mathfrak{su}(2)_R \:. \]
The admissible gauge group is therefore,
\[ 
\G^{\text{adm}}=\U(1)_V\times\SU(2)_L\times\SU(2)_R \:. \]

\subsubsection{Summary of the~$N=2$ Result}
\label{sec:N2summary}

We have found that maximality leaves precisely two possibilities for
admissible gauge groups:
\[ %
\G^{\text{adm}}=
\left\{
\begin{array}{ll}
\SU(2)_V\times\U(1)_L\times\U(1)_R,
& \text{if } V=\R\,\1_2, \\
\U(1)_V\times\SU(2)_L\times\SU(2)_R,
& \text{if } V=\Span_{\R}\{\sigma_1,\sigma_2,\sigma_3\} \:.
\end{array} 
\right. \]

\subsection{The Scalar and Pauli Extensions for General~$N$}
\label{sec:scalar-pauli-extensions}

The two maximal possibilities found for~$N=2$ have natural extensions to higher values of~$N$.  The scalar extension exists for every~$N$, whereas the Pauli extension exists only for even~$N$.  

\subsubsection{The Scalar Extension}
For any~$N \in \N$, one may take the central Clifford subspace
\[ 
\mathscr{A}=\R\,\1_N \:. \]
This is a one-dimensional Clifford space, so it generates~$\bb{C}\ell(1)$.  Since
there is only one Clifford generator,
\[ \bb{C}\ell^{[2]}(\mathscr{A})=0 \:. \]
Moreover,
\[ \Comm(\mathscr{A})=\mathfrak{u}(N) \:, \]
and because~$p=1$ is odd, the twist term may be omitted ($B=0$).  Therefore
\[ \g_A=\R\,\1_N\simeq\mathfrak{u}(1)_A,\qquad\g_V=\mathfrak{u}(N)_V \:. \]
It follows that
\[ \g^{\text{adm}}=\mathfrak{u}(N)_V\oplus\mathfrak{u}(1)_A \:. \]
Splitting the vectorial algebra as
\[ \mathfrak{u}(N)_V=\mathfrak{su}(N)_V\oplus\mathfrak{u}(1)_V \]
and recombining
\[ \mathfrak{u}(1)_V\oplus\mathfrak{u}(1)_A\simeq\mathfrak{u}(1)_L\oplus\mathfrak{u}(1)_R \:, \]
we obtain
\[ \g^{\text{adm}}=\mathfrak{su}(N)_V\oplus\mathfrak{u}(1)_L\oplus\mathfrak{u}(1)_R \:. \]
Thus this family gives the admissible gauge group
\[ 
\G^{\text{adm}}=\SU(N)_V\times \U(1)_L\times \U(1)_R \:. \]

This is the direct generalization of the central~$N=2$ branch.

\subsubsection{The Pauli Extension for Even~$N$}

There is also a direct generalization of the full Pauli branch whenever~$N$ is even.
Write~$N=2r$, and identify
\[ \C^N\simeq \C^2\otimes\C^r \:. \]
We take
\[ 
\mathscr{A}=\Span_{\R}\{\sigma_1\otimes\1_r,\sigma_2\otimes\1_r,\sigma_3\otimes\1_r\} \:. \]
This is the real generating space for~$r$ identical copies of the standard Pauli
representation of~$\bb{C}\ell(3)$.  Since~$p=3$ is odd, we again set~$B=0$.
The Hermitian real form of the bi-vectors is again the Pauli space,
\[ \bb{C}\ell^{[2]}(\mathscr{A})_{\text{Herm}}=\Span_{\R}\{i\sigma_1\sigma_2\otimes\1_r,i\sigma_2\sigma_3\otimes\1_r,i\sigma_3\sigma_1\otimes\1_r\}
=\mathscr{A} \:, \]
up to signs.  The commutant acts on the multiplicity factor,
\[ \Comm(\mathscr{A})=\1_2\otimes\mathfrak{u}(r) \:. \]
Therefore,
\[ \g_A=\mathscr{A}\simeq\mathfrak{su}(2)_A,
\qquad
\g_V\simeq\mathfrak{su}(2)_V\oplus\mathfrak{u}(r)_V \:.  \]
As in the~$N=2$ case, the vectorial and axial Pauli directions can be rewritten
as independent left- and right-handed~$\mathfrak{su}(2)$ generators.  The algebra~$\mathfrak{u}(r)$ remains vectorial.  Hence
\[ \g^{\text{adm}}=\mathfrak{u}(r)_V\oplus\mathfrak{su}(2)_L\oplus\mathfrak{su}(2)_R \:. \]
At the group level, this gives
\[ 
\G^{\text{adm}}=\U(N/2)_V\times\SU(2)_L\times\SU(2)_R \:. \]

This is the direct generalization of the full Pauli~$N=2$ case. Here the Pauli matrices act on the first tensor factor, while the second tensor factor is a multiplicity space. The resulting commutant is therefore enlarged from~$\mathfrak{u}(1)$ to~$\mathfrak{u}(r)$.

\subsubsection{Summary of the Two Families}
The two~$N=2$ branches therefore extend to the following general families:
\[ 
\G^{\text{adm}}
=
\left\{
\begin{array}{ll}
\SU(N)_V\times\U(1)_L\times\U(1)_R,
& \text{if } \mathscr{A}=\R\,\1_N, \\[4pt]
\U(N/2)_V\times\SU(2)_L\times\SU(2)_R,
& \text{if } N=2r \text{ and }
\mathscr{A}=\Span_{\R}\{\sigma_i\otimes\1_r\}_{i=1}^3 \:.
\end{array}
\right. \]
For~$N=4$, these give
\[ \SU(4)_V\times\U(1)_L\times\U(1)_R \]
and
\[ \U(2)_V\times\SU(2)_L\times\SU(2)_R \:. \]
Thus, when we turn to the case~$N=4$, it remains only to determine whether
$\Symm(\C^4)$ admits additional maximal Clifford subspaces beyond these two
general families.

\subsection{The Case~$N=4$}
\label{sec:N4}
We now consider the case of four sectors, thus
\[ \G=\U(4)_L\times\U(4)_R \:. \]
Our task is to classify the maximal Clifford subspaces
\[ \mathscr{A}\subset\Symm(\C^4) \:. \]

The two extensions described in the previous section immediately give two admissible
possibilities. The scalar extension gives
\[ \mathscr{A}=\R\,\1_4 \:, \]
and hence
\[ 
\G^{\text{adm}} =
\SU(4)_V\times\U(1)_L\times\U(1)_R \:. \]
The Pauli extension, with~$\C^4\simeq\C^2\otimes\C^2$, gives
\[ \mathscr{A}=\Span_{\R}\{\sigma_1\otimes\1_2,\sigma_2\otimes\1_2,
\sigma_3\otimes\1_2\} \:, \]
and hence
\[ 
\G^{\text{adm}}
=
\U(2)_V\times\SU(2)_L\times\SU(2)_R \:. \]
These are the direct generalizations of the two maximal~$N=2$ cases. It remains to
check whether~$\Symm(\C^4)$ admits further maximal Clifford subspaces.

\subsubsection{The Irreducible~$\bb{C}\ell(5)$ Case}

In addition to the direct generalizations of the two maximal~$N=2$ cases, there is also a maximal five-dimensional Clifford subspace of~$\Symm(\C^4)$. Let
\[ \Gamma_i=\sigma_1\otimes\sigma_i,\qquad
\Gamma_4=\sigma_2\otimes\1_2,\qquad
\Gamma_5=\sigma_3\otimes\1_2,\qquad i=1,2,3 \:. \]
These matrices are Hermitian and satisfy
\[ \frac{1}{2}\{\Gamma_a,\Gamma_b\}=\delta_{ab}\,\1_4 \:. \]
Thus
\[ \mathscr{A}=\Span_{\R}\{\Gamma_1,\ldots,\Gamma_5\} \]
is the real generating space for the irreducible representation of
$\bb{C}\ell(5)$ on~$\C^4$.

Since~$p=5$ is odd, we may set~$B=0$. The Hermitian real form of the bi-vector space is
\[ \bb{C}\ell^{[2]}(\mathscr{A})_{\text{\rm Herm}}=\Span_{\R}\{i\Gamma_a\Gamma_b:1\leq a<b\leq 5\}\simeq \mathfrak{spin}(5) \]
with Lie bracket~$i[\cdot,\cdot]$. The representation is irreducible because the matrices~$\Gamma_a$ generate all matrices~$\sigma_i\otimes\1_2$ and~$\1_2\otimes\sigma_j$, and hence generate~$M_4(\bb{C})$.  Schur's lemma therefore gives
\[ \Comm_{\C}(\mathscr{A})=\C\1_4 \:. \]
Restricting to Hermitian matrices gives
\[ \Comm(\mathscr{A})=\R\1_4\simeq\mathfrak{u}(1)_V \:. \]

The non-abelian part is generated by the five axial one-vector directions together
with the ten vectorial bi-vector directions. These close under commutation as
\[ \mathscr{A}\oplus\bb{C}\ell^{[2]}(\mathscr{A})_{\text{\rm Herm}}\simeq\mathfrak{spin}(6)\simeq\mathfrak{su}(4) \:. \]
This copy of~$\mathfrak{su}(4)$ is not purely vectorial, nor can it be rewritten as
independent left- and right-handed~$\mathfrak{su}(4)$ factors. An admissible
generator has the form
\[ (A_L,A_R)=(\sigma+\gamma,\sigma-\gamma),\qquad\gamma\in \mathscr{A},\quad
\sigma\in\bb{C}\ell^{[2]}(\mathscr{A})_{\text{\rm Herm}} \:. \]
Thus only the bi-vector subalgebra is vectorial, while the Clifford one-vector directions
are axial. To obtain independent left- and right-handed~$\mathfrak{su}(4)$ factors,
one would need both vectorial and axial copies of the whole~$\mathfrak{su}(4)$, which are not present here. Therefore
\[ \g^{\text{\rm adm}}=\mathfrak{u}(1)_V\oplus\mathfrak{su}(4) \:, \]
where the~$\mathfrak{su}(4)$ factor denotes this mixed embedding. At the group level,
\[ 
\G^{\text{\rm adm}}=\U(1)_V\times\SU(4) \:. \]

\subsubsection{Non-central One-dimensional Clifford Subspaces}
Let~$Q\in\Symm(\C^4)$ be a Hermitian involution, $Q^2=\1_4$. Up to unitary
conjugation, $Q$ has the form
\[ Q=\text{diag}(\1_k,-\1_{4-k}) \:. \]
The cases~$k=0,4$ give the central scalar case. The case~$k=2$ is not maximal, because the two eigenspaces have equal dimension and further Clifford generators can be added. The remaining cases are~$k=1$ and~$k=3$. These are the same up to a change of basis and an overall sign of the generator. We therefore take, without loss of generality,
 \[ Q=\text{diag}(\1_3,-1) \:. \]

The space~$\mathscr{A}=\R Q$ is a one-dimensional Clifford subspace, so
\[ \bb{C}\ell^{[2]}(\mathscr{A})=0 \:. \]
The commutant consists of the Hermitian matrices preserving the~$+1$ and~$-1$ eigen\-spa\-ces of~$Q$, namely
\[ \Comm(\mathscr{A})=\mathfrak{u}(3)\oplus\mathfrak{u}(1) \:. \]
Since~$p=1$ is odd, the twist term may be omitted. Therefore
\[ \g_A=\R\,Q\simeq\mathfrak{u}(1)_A \:,
\qquad
\g_V=
\big(\mathfrak{u}(3)\oplus\mathfrak{u}(1)\big)_V \;. \]
Thus
\[ \g^{\text{adm}}=\big(\mathfrak{u}(3)\oplus\mathfrak{u}(1)\big)_V\oplus\mathfrak{u}(1)_A \:, \]
and at the group level
\[ 
\G^{\text{adm}}=\big(\U(3)\times\U(1)\big)_V\times\U(1)_A \:. \]
Here the axial~$\mathfrak{u}(1)_A$ is generated by~$Q$, not by the identity.  Thus,
although it may be recombined with the corresponding vectorial $Q$-direction, this
does not simply give the scalar factors~$\mathfrak{u}(1)_L\oplus\mathfrak{u}(1)_R$ as
in the central case.

\subsubsection{Summary of the~$N=4$ Result}

Other Clifford subspaces do exist, but they give admissible subalgebras of the maximal cases already listed.  For example, any even-dimensional Clifford subspace can be enlarged by adjoining the pseudoscalar.

The admissible gauge groups are therefore
\[ 
\G^{\text{adm}}=
\left\{
\begin{array}{ll}
\SU(4)_V\times\U(1)_L\times\U(1)_R,
& \text{if } \mathscr{A}=\R\,\1_4, \\[4pt]
\U(2)_V\times\SU(2)_L\times\SU(2)_R,
& \text{if } \mathscr{A}=\Span_{\R}\{\sigma_i\otimes\1_2\}_{i=1}^3, \\[4pt]
\big(\U(3)\times\U(1)\big)_V\times\U(1)_A,
& \text{if } \mathscr{A}=\R\,Q,\quad Q=\text{diag}(\1_3,-1), \\[4pt]
\U(1)_V\times\SU(4),
& \text{if } \mathscr{A}=\Span_{\R}\{\Gamma_1,\ldots,\Gamma_5\} \:.
\end{array}
\right. \]

\Thanks{{{\em{Acknowledgments:}}
We would like to thank Niky Kamran for helpful discussions.
We are grateful to the ``Universit\"atsstiftung Hans Vielberth'' for support.
}}

\bibliographystyle{amsplain}
\providecommand{\bysame}{\leavevmode\hbox to3em{\hrulefill}\thinspace}
\providecommand{\MR}{\relax\ifhmode\unskip\space\fi MR }
\providecommand{\MRhref}[2]{%
  \href{http://www.ams.org/mathscinet-getitem?mr=#1}{#2}
}
\providecommand{\href}[2]{#2}

\end{document}